\documentclass[letterpaper]{easychair}

\usepackage{amsfonts}
\usepackage{amssymb}
\usepackage{amsmath}
\usepackage{stmaryrd}
\usepackage{bbm}
\usepackage{gastex}

\usepackage{hyperref}

\compatiblegastexun

\newcommand{\sat}{\ensuremath{\vDash}}
\newcommand{\allpaths}{\ensuremath{\mathsf{A}}}
\newcommand{\somepath}{\ensuremath{\mathsf{E}}}
\newcommand{\xnext}{\ensuremath{\mathsf{X}}}
\newcommand{\eventually}{\ensuremath{\mathsf{F}}}
\newcommand{\until}{\ensuremath{\mathsf{U}}}
\newcommand{\releases}{\ensuremath{\mathsf{R}}}
\newcommand{\globally}{\ensuremath{\mathsf{G}}}
\newcommand{\Obs}{\mbox{\rm Obs}}

\newcommand{\converges}{\ensuremath{\downarrow}}
\newcommand{\suparrow}[1]{\stackrel{\mbox{\scriptsize $#1$}}{\longrightarrow}}
\newcommand{\Suparrow}[1]{\stackrel{\mbox{\scriptsize $#1$}}{\Longrightarrow}}
\newcommand{\traces}{\mbox{\rm traces}}
\newcommand{\trace}{\mbox{\rm trace}}
\newcommand{\paths}{\ensuremath{\Pi}}
\newcommand{\start}{\mbox{\rm start}}
\newcommand{\actcount}{\mbox{\rm count}}
\newcommand{\refuses}{\mbox{\rm ref}}
\newcommand{\Fin}{\mbox{\rm Fin}}
\newcommand{\init}{\mbox{\rm init}}

\newcommand{\setests}{\ensuremath{\mathcal{ST}}}
\newcommand{\ppre}{\ensuremath{\sqsubseteq}}
\newcommand{\peq}{\ensuremath{\simeq}}
\newcommand{\may}{\mbox{\rm may}}
\newcommand{\must}{\mbox{\rm must}}
\newcommand{\pstop}{\mbox{\rm stop}}
\newcommand{\pass}{\mbox{\rm pass}}
\newcommand{\choice}{\ensuremath{\Box}}
\newcommand{\mchoice}{\ensuremath{\Sigma}}
\newcommand{\synrule}[3]{\frac{\begin{array}{c}#1\end{array}}{\begin{array}{c}#2\end{array}}\ #3}
\newcommand{\parallelth}{\ensuremath{\|_{\theta}}}
\newcommand{\ftr}{\mbox{\rm ftr}}
\newcommand{\seqt}{\mbox{\rm st}}
\newcommand{\ap}{\ensuremath{\mathsf{AP}}}
\newcommand{\natset}{\ensuremath{\mathbbm{N}}}
\newcommand{\bnfor}{\ \vert\ }
\newcommand{\stablefail}{\ensuremath{\mathcal{SF}}}
\newcommand{\nottest}[1]{\ensuremath{\overline{#1}}}

\newcommand{\Lor}{\bigvee}
\newcommand{\restrict}[2]{\ensuremath{[#1]_{#2}}}
\newcommand{\ltstokripke}{\ensuremath{\mathbbm{K}}}  
\newcommand{\ltstokripkex}{\ensuremath{\mathbbm{X}}}
\newcommand{\fttoctl}{\ensuremath{\mathbbm{F}}}      
\newcommand{\ctltoft}{\ensuremath{\mathbbm{T}}}      
\newcommand{\internal}{\ensuremath{\mathbf{i}}}
\newcommand{\procset}{\ensuremath{\mathcal{P}}}
\newcommand{\ftset}{\ensuremath{\mathcal{T}}}
\newcommand{\ctlset}{\ensuremath{\mathcal{F}}}

\newtheorem{definition}{Definition}
\newtheorem{theorem}{Theorem}
\newtheorem{corollary}[theorem]{Corollary}
\newtheorem{proposition}[theorem]{Proposition}
\newtheorem{lemma}[theorem]{Lemma}
\newenvironment{proof*}[1]
{\noindent \emph{Proof of #1.} \rm} {\hfill \qed}

\newtheorem{example}{Example}
\newenvironment{qexample}[1]
{
 \begin{quotation}
 \noindent
 \hrulefill
 \begin{example}
 \textsc{#1:}\rm\\
}
{
 \end{example}
  \par
 \noindent
 \hrulefill
 \end{quotation}
}

\newcommand{\nudge}{\mbox{\rm nudge}}
\newcommand{\turn}{\mbox{\rm turn}}
\newcommand{\coin}{\mbox{\rm coin}}
\newcommand{\bang}{\mbox{\rm bang}}
\newcommand{\tea}{\mbox{\rm tea}}
\newcommand{\coffee}{\mbox{\rm coffee}}
\newcommand{\water}{\mbox{\rm water}}

\def\today{\number\day\space\ifcase\month\or January\or February\or
  March\or April\or May\or June\or July\or August\or September\or
  October\or November\or December\fi\space\number\year}

\title{A Constructive Equivalence between Computation Tree Logic and
  Failure Trace Testing\thanks{Part of this research was supported by
    the Natural Sciences and Engineering Research Council of Canada
    and by Bishop's University.}}
\author{Stefan D.\ Bruda, Sunita Singh, A.\ F.\ M.\ Nokib Uddin, Zhiyu Zhang,  and Rui Zuo\\
  Department of Computer Science\\
  Bishop's University\\
  Sherbrooke, Quebec J1M 1Z7, Canada\\
  stefan@bruda.ca,
  $\{$singh$|$uddin$|$zzhang$|$zuo$\}$@cs.ubishops.ca} \date{\today}

\titlerunning{}
\authorrunning{}

\begin{document}

\maketitle

\begin{abstract}
  The two major systems of formal verification are model checking and
  algebraic model-based testing. Model checking is based on some form
  of temporal logic such as linear temporal logic (LTL) or computation
  tree logic (CTL). One powerful and realistic logic being used is
  CTL, which is capable of expressing most interesting properties of
  processes such as liveness and safety.  Model-based testing is based
  on some operational semantics of processes (such as traces,
  failures, or both) and its associated preorders. The most
  fine-grained preorder beside bisimulation (mostly of theoretical
  importance) is based on failure traces.  We show that these two most
  powerful variants are equivalent; that is, we show that for any
  failure trace test there exists a CTL formula equivalent to it, and
  the other way around. All our proofs are constructive and
  algorithmic.  Our result allows for parts of a large system to be
  specified logically while other parts are specified algebraically,
  thus combining the best of the two (logic and algebraic) worlds.

  \textbf{Keywords:} model checking, model-based testing, stable
  failure, failure trace, failure trace preorder, temporal logic,
  computation tree logic, labelled transition system, Kripke structure
\end{abstract}

\section{Introduction}

Computing systems are already ubiquitous in our everyday life, from
entertainment systems at home, to telephone networks and the Internet,
and even to health care, transportation, and energy infrastructure.
Ensuring the correct behaviour of software and hardware has been one
of the goals of Computer Science since the dawn of computing.  Since
then computer use has skyrocketed and so has the need for assessing
correctness.

Historically the oldest verification method, which is still widely
used today, is empirical testing \cite{pawlikowski90,schriber03}.
This is a non-formal method which provides input to a system, observes
the output, and verifies that the output is the one expected given the
input.  Such testing cannot check all the possible input combinations
and so it can disprove correctness but can never prove it.  Deductive
verification \cite{floyd67,hoare69,saidi97} is chronologically the
next verification method developed.  It consists of providing proofs
of program correctness manually, based on a set of axioms and
inference rules.  Program proofs provide authoritative evidence of
correctness but are time consuming and require highly qualified
experts.

Various techniques have been developed to automatically perform
program verification with the same effect as deductive reasoning but
in an automated manner.  These efforts are grouped together in the
general field of formal methods.  The general technique is to verify a
system automatically against some formal specification.  Model-based
testing and model checking are the two approaches to formal methods
that became mainstream.  Their roots can be traced to simulation and
deductive reasoning, respectively. These formal methods however are
sound, complete, and to a large extent automatic.  They have proven
themselves through the years and are currently in wide use throughout
the computing industry.

In model-based testing \cite{motres05,denicola84,tretmans96} the
specification of a system is given algebraically, with the underlying
semantics given in an operational manner as a labeled transition
system (LTS for short), or sometimes as a finite automaton (a
particular, finite kind of LTS).  Such a specification is usually an
abstract representation of the system's desired behaviour.  The system
under test is modeled using the same formalism (either finite or
infinite LTS).  The specification is then used to derive
systematically and formally tests, which are then applied to the
system under test.  The way the tests are generated ensures soundness
and completeness.  In this paper we focus on arguably the most
powerful method of model-based testing, namely failure trace testing
\cite{langerak89}.  Failure trace testing also introduces a smaller
set (of sequential tests) that is sufficient to assess the failure
trace relation.

By contrast, in model checking \cite{clarke86,clarke99,queille83} the
system specification is given in some form of temporal logic.  The
specification is thus a (logical) description of the desired
properties of the system.  The system under test is modeled as Kripke
structures, another formalism similar to transition systems.  The
model checking algorithm then determines whether the initial states of
the system under test satisfy the specification formulae, in which
case the system is deemed correct.  There are numerous temporal logic
variants used in model checking, including CTL*, CTL and LTL. In this
paper we focus on CTL.

There are advantages as well as disadvantages to each of these formal
methods techniques.  Model checking is a complete verification
technique, which has been widely studied and also widely used in
practice.  The main disadvantage of this technique is that it is not
compositional.  It is also the case that model checking is based on
the system under test being modeled using a finite state formalism,
and so does not scale very well with the size of the system under
test.  By contrast, model-based testing is compositional by definition
(given its algebraic nature), and so has better scalability.  In
practice however it is not necessarily complete given that some of the
generated tests could take infinite time to run and so their success
or failure cannot be readily ascertained.  The logical nature of
specification for model checking allows us to only specify the
properties of interest, in contrast with the labeled transition
systems or finite automata used in model-based testing which more or
less require that the whole system be specified.

Some properties of a system may be naturally specified using temporal
logic, while others may be specified using finite automata or labeled
transition systems.  Such a mixed specification could be given by
somebody else, but most often algebraic specifications are just more
convenient for some components while logic specifications are more
suitable for others.  However, such a mixed specification
cannot be verified.  Parts of it can be model checked and some other
parts can be verified using model-based testing.  However, no global
algorithm for the verification of the whole system exists.  Before
even thinking of verifying such a specification we need to convert one
specification to the form of the other.

We describe in this paper precisely such a conversion.  We first
propose two equivalence relations between labeled transition systems
(the semantic model used in model-based testing) and Kripke structures
(the semantic model used in model checking), and then we show that for
each CTL formula there exists an equivalent failure trace test suite,
and the other way around.  In effect, we show that the two (algebraic
and logic) formalisms are equivalent.  All our proofs are constructive
and algorithmic, so that implementing back and forth automated
conversions is an immediate consequence of our result.

We believe that we are thus opening the domain of combined, algebraic
and logic methods of formal system verification.  The advantages of
such a combined method stem from the above considerations but also
from the lack of compositionality of model checking (which can thus be
side-stepped by switching to algebraic specifications), from the lack
of completeness of model-based testing (which can be side-stepped by
switching to model checking), and from the potentially attractive
feature of model-based testing of incremental application of a test
suite insuring correctness to a certain degree (which the
all-or-nothing model-checking lacks).

The reminder of this paper is organized as follows: We introduce basic
concepts including model checking, temporal logic, model-based
testing, and failure trace testing in the next section.  Previous work
is reviewed briefly in Section~\ref{sect-previous}.
Section~\ref{sect-lts-kripke} defines the concept of equivalence
between LTS and Kripke structures, together with an algorithmic
function for converting an LTS into its equivalent Kripke structure.
Two such equivalence relations and conversion functions are offered
(Section~\ref{sect-lts-kripke-a} and~\ref{sect-lts-kripke-b},
respectively).  Section~\ref{sect-eq-ctl-ft} then presents our
algorithmic conversions from failure trace tests to CTL formulae
(Section~\ref{sect-ft-to-ctl} with an improvement in
Section~\ref{sect-cycle}) and the other way around
(Section~\ref{sect-ctl-to-ft}).  We discuss the significance and
consequences of our work in Section~\ref{sect-conclusions}.  For the
remainder of this paper results proved elsewhere are introduced as
Propositions, while original results are stated as Theorems, Lemmata,
or Corollaries.

\section{Preliminaries}
\label{sect-prelims}

This section is dedicated to introducing the necessary background
information on model checking, temporal logic, and failure trace
testing.  For technical reasons we also introduce TLOTOS, a process
algebra used for describing algebraic specifications, tests, and
systems under test.  The reason for using this particular language is
that earlier work on failure trace testing uses this language as well.

Given a set of symbols $A$ we use as usual $A^{*}$ to denote exactly
all the strings of symbols from $A$.  The empty string, and only the
empty string is denoted by $\varepsilon$.  We use $\omega$ to refer to
$|\natset|$, the cardinality of the set $\natset$ of natural numbers.
The power set of a set $A$ is denoted as usual by $2^{A}$.

\subsection{Temporal Logic and Model Checking}

A specification suitable for model checking is described by a temporal
logic formula.  The system under test is given as a Kripke structure.
The goal of model checking is then to find the set of all states in
the Kripke structure that satisfy the given logic formula.  The system
then satisfies the specification provided that all the designated
initial states of the respective Kripke structure satisfy the logic
formula.

Formally, a \emph{Kripke structure} \cite{clarke99} $K$ over a set
$\ap$ of atomic propositions is a tuple $(S, S_{0}, \rightarrow, L)$,
where $S$ is a set of states, $S_{0}\subseteq S$ is the set of initial
states, $\rightarrow\subseteq S\times S$ is the transition relation,
and $L:S\rightarrow 2^{\ap}$ is a function that assigns to each states
exactly all the atomic propositions that are true in that state.  As
usual we write $s \rightarrow t$ instead of $(s,t)\in \rightarrow$.
It is usually assumed \cite{clarke99} that $\rightarrow$ is total,
meaning that for every state $s\in S$ there exists a state $t\in S$
such that $s\rightarrow t$.  Such a requirement can however be easily
established by creating a ``sink'' state that has no atomic
proposition assigned to it, is the target of all the transitions from
states with no other outgoing transitions, and has one outgoing
``self-loop'' transition back to itself.

A \emph{path} $\pi$ in a Kripke structure is a sequence
$s_{0}\rightarrow s_{1}\rightarrow s_{2} \rightarrow \cdots$ such that
$s_{i}\rightarrow s_{i+1}$ for all $i\geq 0$. The path starts from
state $s_{0}$. Any state may be the start of multiple paths.  It
follows that all the paths starting from a given state $s_{0}$ can be
represented together as a computation tree with nodes labeled with
states.  Such a tree is rooted at $s_{0}$ and $(s,t)$ is an edge in
the tree if and only if $s \rightarrow t$.  Some temporal logics
reason about computation paths individually, while some other temporal
logics reason about whole computation trees.

There are several temporal logics currently in use.  We will focus in
this paper on the CTL* family \cite{clarke99,nicola95} and more
precisely on the CTL variant.  CTL* is a general temporal logic which
is usually restricted for practical considerations.  One such a
restriction is the linear-time temporal logic or LTL
\cite{clarke99,pnueli81}, which is an example of temporal logic that
represents properties of individual paths.  Another restriction is the
computation tree logic or CTL \cite{clarke82,clarke99}, which
represents properties of computation trees.

In CTL* the properties of individual paths are represented using five
temporal operators: \xnext\ (for a property that has to be true in the
next state of the path), \eventually\ (for a property that has to
eventually become true along the path), \globally\ (for a property
that has to hold in every state along the path), \until\ (for a
property that has to hold continuously along a path until another
property becomes true and remains true for the rest of the path), and
\releases\ (for a property that has to hold along a path until another
property becomes true and releases the first property from its
obligation).  These path properties are then put together so that they
become state properties using the quantifiers \allpaths\ (for a
property that has to hold on all the outgoing paths) and \somepath\
(for a property that needs to hold on at least one of the outgoing
paths).

CTL is a subset of CTL*, with the additional restriction that the
temporal constructs \xnext, \eventually, \globally, \until, and
\releases\ must be immediately preceded by one of the path quantifiers
\allpaths\ or \somepath.  More precisely, the syntax of CTL formulae
is defined as follows:
\begin{eqnarray*}
f & = & \top \bnfor \bot \bnfor a \bnfor \lnot f \bnfor f_{1} \land
f_{2} \bnfor f_{1}\lor f_{2} \bnfor \\
&& \allpaths\xnext\ f \bnfor
\allpaths\eventually\ f \bnfor \allpaths\globally\ f \bnfor \allpaths\
f_{1}\ \until\ f_{2} \bnfor \allpaths\ f_{1}\ \releases\ f_{2} \bnfor\\
&&
\somepath\xnext\ f \bnfor \somepath\eventually\ f \bnfor
\somepath\globally\ f \bnfor\ \somepath\ f_{1}\ \until\ f_{2} \bnfor
\somepath\ f_{1}\ \releases\ f_{2}
\end{eqnarray*}
where $a\in \ap$, and $f$, $f_{1}$, $f_{2}$ are all state formulae.

CTL formulae are interpreted over states in Kripke structures.
Specifically, the CTL semantics is given by the operator $\sat$ such
that $K, s \sat f$ means that the formula $f$ is true in the state $s$
of the Kripke structure $K$. All the CTL formulae are state formulae,
but their semantics is defined using the intermediate concept of path
formulae.  In this context the notation $K, \pi \sat f$ means that the
formula $f$ is true along the path $\pi$ in the Kripke structure $K$.
The operator \sat\ is defined inductively as follows:
\begin{enumerate}
\item $K, s\sat \top$ is true and $K, s\sat \bot$ is false for any
  state $s$ in any Kripke structure $K$.
\item $K, s\sat a$, $a\in \ap$ if and only if $a\in L(s)$.
\item $K, s\sat\lnot f$ if and only if $\lnot (K,s\sat f)$ for any
  state formula $f$.
\item $K, s\sat f\land g$ if and only if $K, s\sat f$ and $K, s\sat g$
  for any state formulae $f$ and $g$.
\item $K, s\sat f\lor g$ if and only if $K, s\sat f$ or $K, s\sat g$
  for any state formulae $f$ and $g$.
\item $K, s\sat \somepath\ f$ for some path formula $f$ if and only if
  there exists a path $\pi = s \rightarrow s_{1} \rightarrow
  s_{2}\rightarrow \cdots \rightarrow s_{i}$, $i\in
  \natset\cup\{\omega\}$ such that $K, \pi \sat f$.
\item $K, s\sat \allpaths\ f$ for some path formula $f$ if and only if
  $K, \pi \sat f$ for all paths $\pi = s \rightarrow s_{1} \rightarrow
  s_{2}\rightarrow \cdots \rightarrow s_{i}$, $i\in
  \natset\cup\{\omega\}$.
\end{enumerate}

We use $\pi^{i}$ to denote the $i$-th state of a path $\pi$, with the
first state being $\pi^{0}$.  The operator \sat\ for path formulae is
then defined as follows:
\begin{enumerate}
\item $K,\pi\sat \xnext\ f$ if and only if $k,\pi^{1}\sat f$ for any state
  formula $f$.
\item $K,\pi\sat f\ \until\ g$ for any state formulae $f$ and $g$ if
  and only if there exists $j\geq 0$ such that $K,\pi^{k}\sat g$ for
  all $k \geq j$, and $K,\pi^{i}\sat f$ for all $i<j$.
\item $K,\pi\sat f\ \releases\ g$ for any state formulae $f$ and $g$
  if and only if for all $j\geq 0$, if $K,\pi^{i}\not\sat f$ for every
  $i<j$ then $K,\pi^{j}\sat g$.
\end{enumerate}

\subsection{Labeled Transition Systems and Stable Failures}

CTL semantics is defined over Kripke structures, where each state is
labeled with atomic propositions. By contrast, the common model used
for system specifications in model-based testing is the labeled
transition system (LTS), where the labels (or actions) are associated
with the transitions instead.

An LTS \cite{katoen05} is a tuple $M = (S, A, \rightarrow, s_{0})$
where $S$ is a countable, non empty set of states, $s_{0}\in S$ is the
initial state, and $A$ is a countable set of actions.  The actions in
$A$ are called visible (or observable), by contrast with the special,
unobservable action $\tau\not\in A$ (also called internal action).
The relation $\rightarrow \subseteq S\times (A\cup\{\tau\})\times S$
is the transition relation; we use $p\suparrow{a}q$ instead of $(p, a,
q)\in \rightarrow$.  A transition $p\suparrow{a}q$ means that state
$p$ becomes state $q$ after performing the (visible or internal)
action $a$.

The notation $p\suparrow{a}$ stands for $\exists p':p\suparrow{a}p'$.
The sets of states and transitions can also be considered global, in
which case an LTS is completely defined by its initial state.  We
therefore blur whenever convenient the distinction between an LTS and
a state, calling them both ``processes''.  Given that $\rightarrow$ is
a relation rather than a function, and also given the existence of the
internal action, an LTS defines a nondeterministic process.

A \emph{path} (or \emph{run}) $\pi$ starting from state $p'$ is a
sequence $p' = p_{0}\suparrow{a_{1}} p_{1}\suparrow{a_{2}} \cdots
p_{k-1} \suparrow{a_{k}} p_{k}$ with $k\in\natset\cup\{\omega\}$ such
that $p_{i-1}\suparrow{a_{i}}p_{i}$ for all $0<i\leq k$.  We use
$|\pi|$ to refer to $k$, the length of $\pi$.  If $|\pi|\in \natset$,
then we say that $\pi$ is finite. The trace of $\pi$ is the sequence
$\trace(\pi) = (a_{i})_{0< i\leq |\pi|, a_{i}\neq \tau}\in A^{*}$ of
all the visible actions that occur in the run listed in their order of
occurrence and including duplicates.  Note in particular that internal
actions do not appear in traces.  The set of finite traces of a
process $p$ is defined as $\Fin(p) = \{tr \in \traces(p) : |tr|\in
\natset \}$.  If we are not interested in the intermediate states of a
run then we use the notation $p\Suparrow{w}q$ to state that there
exists a run $\pi$ starting from state $p$ and ending at state $q$
such that $\trace(\pi) = w$.  We also use $p\Suparrow{w}$ instead of
$\exists p':p\Suparrow{w}p'$.

A process $p$ that has no outgoing internal action cannot make any
progress unless it performs a visible action.  We say that such a
process is \emph{stable} \cite{schneider00}.  We write $p\converges$
whenever we want to say that process $p$ is stable.  Formally,
$p\converges = \lnot (\exists p'\neq p: p\Suparrow{\varepsilon}p')$.
A stable process $p$ responds predictably to any set of actions
$X\subseteq A$, in the sense that its response depends exclusively on
its outgoing transitions.  Whenever there is no action $a\in X$ such
that $p\suparrow{a}$ we say that $p$ \emph{refuses} the set $X$.  Only
stable processes are able to refuse actions; unstable processes refuse
actions ``by proxy'': they refuse a set $X$ whenever they can
internally become a stable process that refuses $X$.  Formally, $p$
refuses $X$ (written $p\ \refuses\ X$) if and only if $\forall a\in X:
\lnot(\exists p': (p\Suparrow{\varepsilon }p') \land p'\converges\land
p'\suparrow{a})$.

To describe the behaviour of a process in terms of refusals we need to
record each refusal together with the trace that causes that refusal.
An observation of a refusal plus the trace that causes it is called a
\emph{stable failure} \cite{schneider00}.  Formally, $(w, X)$ is a
stable failure of process $p$ if and only if $\exists p^{w}:
(p\Suparrow{w}p^{w})\land p^{w}\converges \land (p^{w}\ \refuses\ X)$.
The set of stable failures of $p$ is then $\stablefail(p) = \{(w,
X):\exists p^{w}: (p\Suparrow{w}p^{w})\land p^{w}\converges \land
(p^{w}\ \refuses\ X)\}$.

Several preorder relations (that is, binary relations that are
reflexive and transitive but not necessarily symmetric or
antisymmetric) can be defined over processes based on their observable
behaviour (including traces, refusals, stable failures, etc.)
\cite{bruda05}.  Such preorders can then be used in practice as
implementation relations, which in turn create a process-oriented
specification technique.  The \emph{stable failure preorder} is
defined based on stable failures and is one of the finest such
preorders (but not the absolute finest) \cite{bruda05}.

Let $p$ and $q$ be two processes.  The stable failure preorder
$\ppre_{\mathrm{SF}}$ is defined as $p\ppre_{\mathrm{SF}} q$ if and only if
$\Fin(p) \subseteq \Fin(q)$ and $\stablefail(p) \subseteq
\stablefail(q)$.  Given the preorder $\ppre_{\mathrm{SF}}$ one can
naturally define the stable failure equivalence $\peq_{\stablefail}$:
$p\peq_{\stablefail}q$ if and only if $p\ppre_{\mathrm{SF}}q$ and
$q\ppre_{\mathrm{SF}}p$.

\subsection{Failure Trace Testing}
\label{sect-ft}

In model-based testing \cite{motres05} a test runs in parallel with
the system under test and synchronizes with it over visible actions. A
run of a test $t$ and a process $p$ represents a possible sequence of
states and actions of $t$ and $p$ running synchronously.  The outcome
of such a run is either success ($\top$) or failure ($\bot$).  The
precise definition of synchronization, success, and failure depends on
the particular type of tests being considered. We will present below
such definitions for the particular framework of failure trace
testing.

Given the nondeterministic nature of LTS there may be multiple runs
for a given process $p$ and a given test $t$ and so a set of outcomes
is necessary to give the results of all the possible runs.  We denote
by $\Obs(p, t)$ the set of exactly all the possible outcomes of all
the runs of $p$ and $t$.  Given the existence of such a set of
outcomes, two definitions of a process passing a test are possible.
More precisely, a process $p$ \emph{may} pass a test $t$ whenever some
run is successful (formally, $p\ \may\ t$ if and only if $\top\in \Obs(p, t)$),
while $p$ \emph{must} pass $t$ whenever all runs are successful
(formally, $p\ \must\ t$ if and only if $\{\top\} = \Obs(p, t)$).

In what follows we use the notation $\init(p) = \{a\in A:
p\Suparrow{a}\}$.  A failure trace $f$ \cite{langerak89} is a string
of the form $f=A_{0}a_{1}A_{1}a_{2}A_{2}\ldots a_{n}A_{n}$, $n\geq 0$,
with $a_{i}\in A^{*}$ (sequences of actions) and $A_{i}\subseteq A$
(sets of refusals).  Let $p$ be a process such that
$p\Suparrow{\varepsilon}p_{0}\Suparrow{a_{1}}p_{1}\Suparrow{a_{2}}\cdots
\Suparrow{a_{n}}p_{n}$; $f=A_{0}a_{1}A_{1}a_{2}A_{2}\ldots a_{n}A_{n}$
is then a failure trace of $p$ whenever the following two conditions
hold:
\begin{itemize}
\item If $\lnot (p_{i}\suparrow\tau)$, then $A_{i}\subseteq
  (A\setminus\init(p_{i}))$; for a stable state the failure trace
  refuses any set of events that cannot be performed in that state
  (including the empty set).
\item If $p_{i}\suparrow\tau$ then $A_{i} = \emptyset$; whenever
  $p_{i}$ is not a stable state it refuses an empty set of events by
  definition.
\end{itemize}
In other words, we obtain a failure trace of $p$ by taking a trace of
$p$ and inserting refusal sets after stable states.

Systems and tests can be concisely described using the testing
language TLOTOS \cite{brinksma87,langerak89}, which will also be used
in this paper.  $A$ is the countable set of observable actions, ranged
over by $a$. The set of processes or tests is ranged over by $t$,
$t_{1}$ and $t_{2}$, while $T$ ranges over the sets of tests or
processes.  The syntax of TLOTOS is then defined as follows:
\[
t = \pstop \bnfor a; t_{1} \bnfor \internal; t_{1} \bnfor \theta;
t_{1} \bnfor \pass \bnfor t_{1}\ \choice\ t_{2} \bnfor \mchoice T
\]
The semantics of TLOTOS is then the following:
\begin{enumerate}
\item inaction (\pstop): no rules.
\item action prefix: $a; t_{1}\suparrow{a} t_{1}$ and $\internal;
  t_{1}\suparrow{\tau} t_{1}$
\item deadlock detection: $\theta; t_{1} \suparrow{\theta} t_{1}$.
\item successful termination: $\pass\suparrow{\gamma}\pstop$.
\item choice: with $g\in A\cup \{\gamma, \theta, \tau \}$,
  \[
  \synrule{t_{1}\suparrow{g}t_{1}'}
  {t_{1}\ \choice\ t_{2}\suparrow{g}t_{1}'\\t_{2}\ \choice\ t_{1}\suparrow{g}t_{1}'}{}
  \]
\item generalized choice: with $g\in A\cup \{\gamma, \theta, \tau
  \}$,
  \[
  \synrule{t_{1}\suparrow{g}t_{1}'}
  {\mchoice (\{t_{1}\}\cup t)\suparrow{g}t_{1}'}{}
  \]
\end{enumerate}

Failure trace tests are defined in TLOTOS using the special actions
$\gamma$ which signals the successful completion of a test, and
$\theta$ which is the deadlock detection label (the precise behaviour
will be given later).  Processes (or LTS) can also be described as
TLOTOS processes, but such a description does not contain $\gamma$
or $\theta$.  A test runs in parallel with the system under test
according to the parallel composition operator
\parallelth.  This operator also defines the semantics of $\theta$ as
the lowest priority action:
\[
\begin{array}{ccc}
  \begin{array}{cc}
  \synrule{p\suparrow{\tau}p'}
  {p\parallelth t\suparrow{\tau}p' \parallelth t}{}
  &
  \synrule{t\suparrow{\tau}t'}
  {p\parallelth t\suparrow{\tau}p' \parallelth t}{}
  \end{array}
  \\\\
  \begin{array}{cc}
    \setlength{\arraycolsep}{0pt}
    \synrule{t\suparrow{\gamma}\pstop}
    {p\parallelth t\suparrow{\gamma}\pstop}{}
    &
    \synrule{p\suparrow{a}p' \qquad t\suparrow{a}t'}
    {p\parallelth t\suparrow{a}p' \parallelth t'}
    { a\in A }
  \end{array}
  \\\\
  \begin{array}{cc}
  \synrule{t\suparrow{\theta}t' \qquad
    \lnot\exists x \in A\cup\{\tau,\gamma\}:p\parallelth t\suparrow{x}}
  {p\parallelth t\suparrow{\theta}p \parallelth t'}
  {}
  \end{array}
\end{array}
\]

Given that both processes and tests can be nondeterministic we have a
set $\paths(p\parallelth t)$ of possible runs of a process and a test.
The outcome of a particular run $\pi\in\paths(p\parallelth t)$ of a
test $t$ and a process under test $p$ is success ($\top$) whenever the
last symbol in $\trace(\pi)$ is $\gamma$, and failure ($\bot$)
otherwise.  One can then distinguish the possibility and the
inevitability of success for a test as mentioned earlier: $p\ \may\ t$
if and only if $\top \in \Obs(p,t)$, and $p\ \must\ t$ if and only if
$\{\top\} = \Obs(p,t)$.

The set \setests\ of sequential tests is defined as follows
\cite{langerak89}: $\pass\in\setests$, if $t\in\setests$ then $a; t\in
\setests$ for any $a\in A$, and if $t\in\setests$ then
$\mchoice\{a;\pstop: a\in A'\}\ \choice\ \theta;t\in\setests$ for any
$A'\subseteq A$.

A bijection between failure traces and sequential tests exists
\cite{langerak89}.  For a sequential test $t$ the failure trace
$\ftr(t)$ is defined inductively as follows: $\ftr(\pass)= \emptyset$,
$\ftr(a; t')= a\ \ftr(t')$, and $\ftr(\mchoice\{a;\pstop: a\in A'\}\
\choice\ \theta;t') = A'\ \ftr(t')$.  Conversely, let $f$ be a failure
trace.  Then we can inductively define the sequential test $\seqt(f)$
as follows: $\seqt(\emptyset) = \pass$, $\seqt(af) = a\ \seqt(f)$, and
$\seqt(A f) = \mchoice\{a;\pstop: a\in A\}\ \choice\ \theta;
\seqt(f)$.  For all failure traces $f$ we have that $\ftr(\seqt(f)) =
f$, and for all tests $t$ we have $\seqt(\ftr(t)) = t$. We then define
the failure trace preorder $\ppre_{\mathrm{FT}}$ as follows:
$p\ppre_{\mathrm{FT}}q$ if and only if $\ftr(p)\subseteq \ftr(q)$.

The above bijection effectively shows that the failure trace preorder
(which is based on the behaviour of processes) can be readily
converted into a testing-based preorder (based on the outcomes of
tests applied to processes).  Indeed there exists a successful run of
$p$ in parallel with the test $t$, if and only if $f$ is a failure trace of both
$p$ and $t$.  Furthermore, these two preorders are equivalent to the
stable failure preorder introduced earlier:

\begin{proposition}
  \label{th-stable-failure-trace}
  {\rm\textbf{\cite{langerak89}}} Let $p$ be a process, $t$ a
  sequential test, and $f$ a failure trace.  Then $p\ \may\ t$ if and only if $f
  \in \ftr(p)$, where $f=\ftr(t)$.

  Let $p_{1}$ and $p_{2}$ be processes.  Then $p_{1}
  \ppre_{\mathrm{SF}} p_{2}$ if and only if $p_{1} \ppre_{\mathrm{FT}} p_{2}$ if and only if
  $p_{1}\ \may\ t \implies p_{2}\ \may\ t$ for all failure trace tests
  $t$ if and only if $\forall t'\in \setests: p_{1}\ \may\ t' \implies p_{2}\
  \may\ t'$.

  Let $t$ be a failure trace test.  Then there exists $T(t) \subseteq
  \setests$ such that $p\ \may\ t $ if and only if $\exists t'\in T(t): p\ \may\
  t'$.
\end{proposition}

We note in passing that unlike other preorders, $\ppre_{\mathrm{SF}}$
(or equivalently $\ppre_{\mathrm{FT}}$) can be in fact characterized
in terms of may testing only; the must operator needs not be
considered any further.

\section{Previous Work}
\label{sect-previous}

The investigation into connecting logical and algebraic frameworks of
formal specification and verification has not been pursued in too much
depth.  To our knowledge the only substantial investigation on the
matter is based on linear-time temporal logic (LTL) and its relation
with B\"{u}chi automata \cite{thomas90}.  Such an investigation
started with timed B\"{u}chi automata \cite{alur94} approaches to LTL
model checking \cite{clarke99,denicola84,gerth95,vardi86,vardi94}.

An explicit equivalence between LTL and the may and must testing
framework of De~Nicola and Hennessy \cite{denicola84} was developed as
a unified semantic theory for heterogeneous system specifications
featuring mixtures of labeled transition systems and LTL formulae
\cite{cleaveland00}.  This theory uses B\"{u}chi automata
\cite{thomas90} rather than LTS as underlying semantic formalism.  The
B\"{u}chi must-preorder for a certain class of B\"{u}chi process was
first established by means of trace inclusion.  Then LTL formulae
were converted into B\"{u}chi processes whose languages contain the
traces that satisfy the formula.

The relation between may and must testing and temporal logic mentioned
above \cite{cleaveland00} was also extended to the timed (or
real-time) domain \cite{bruda10j,dai08}.  Two refinement timed
preorders similar to may and must testing were introduced, together
with behavioural and language-based characterizations for these
relations (to show that the new preorders are extensions of the
traditional preorders).  An algorithm for automated test generation
out of formulae written in a timed variant of LTL called Timed
Propositional Temporal Logic (TPTL) \cite{bellini00} was then
introduced.

To our knowledge there were only two efforts on the equivalence
between CTL and algebraic specifications \cite{bruda09a,nicola90}, one
of which is our preliminary version of this paper.  This earlier
version \cite{bruda09a} presents the equivalence between LTS and
Kripke structures (Section~\ref{sect-lts-kripke} below) and also a
tentative (but not complete) conversion from failure trace tests to
CTL formulae.  The other effort \cite{nicola90} is the basis of our
second equivalence relation between LTS and Kripke structure (again
see Section~\ref{sect-lts-kripke} below).

\section{Two Constructive Equivalence Relations between LTS and Kripke Structures}
\label{sect-lts-kripke}

We believe that the only meaningful basis for constructing a Kripke
structure equivalent to a given LTS is by taking the outgoing actions
of an LTS state as the propositions that hold on the equivalent Kripke
structure state.  This idea is amenable to at least two algorithmic
conversion methods.

\subsection{Constructing a Compact Kripke Structure Equivalent with a Given LTS}
\label{sect-lts-kripke-a}

We first define an LTS satisfaction operator similar to the one on
Kripke structures in a natural way (and according to the intuition
presented above).

\begin{definition}
  \label{def-lts-sat}
  \textsc{Satisfaction for processes:} A process $p$ satisfies $a\in
  A$, written by abuse of notation $p\sat f$, iff $p\suparrow{a}$.
  That $p$ satisfies some (general) CTL* state formula is defined
  inductively as follows: Let $f$ and $g$ be some state formulae
  unless stated otherwise; then,

  \begin{enumerate}
  \item $p\sat \top$ is true and $p\sat \bot$ is false for any process
    $p$.
  \item $p\sat\lnot f$ iff $\lnot (p\sat f)$.
  \item $p\sat f\land g$ iff $p\sat f$ and $p\sat g$.
  \item $p\sat f\lor g$ iff $p\sat f$ or $p\sat g$.
  \item $p\sat \somepath\ f$ for some path formula $f$ iff there is a
    path $\pi = p\suparrow{a_{0}} s_{1}\suparrow{a_{1}} s_{2}
    \suparrow{a_{2}} \cdots$ such that $\pi \sat f$.
  \item $p\sat \allpaths\ f$ for some path formula $f$ iff $p \sat f$
    for all paths $\pi = p\suparrow{a_{0}} s_{1}\suparrow{a_{1}} s_{2}
    \suparrow{a_{2}} \cdots$.
  \end{enumerate}

  We use $\pi^{i}$ to denote the $i$-th state of a path $\pi$ (with
  the first state being state 0, or $\pi^{0}$).  The definition of
  \sat\ for LTS paths is:

  \begin{enumerate}
  \item $\pi\sat \xnext\ f$ iff $\pi^{1}\sat f$.
  \item $\pi\sat f\ \until\ g$ iff there exists $j\geq 0$ such that
    $\pi^{j}\sat g$ and $\pi^{k}\sat g$ for all $k \geq j$, and
    $\pi^{i}\sat f$ for all $i<j$.
  \item $\pi\sat f\ \releases\ g$ iff for all $j\geq 0$, if
    $\pi^{i}\not\sat f$ for every $i<j$ then $\pi^{j}\sat g$.
  \end{enumerate}
\end{definition}

We also need to define a weaker satisfaction operator for CTL.  Such
an operator is similar to the original, but is defined over a set of
states rather than a single state.  By abuse of notation we denote
this operator by $\sat$ as well.

\begin{definition}
  \label{def-kripke-set-sat}
  \textsc{Satisfaction over sets of states:} Consider a Kripke
  structure $K=(S, S_{0}, R, L)$ over $\ap$.  For some set $Q\subseteq
  S$ and some CTL state formula $f$ we define $K,Q\sat f$ as follows,
  with $f$ and $g$ state formulae unless stated otherwise:
  \begin{enumerate}
  \item $K, Q\sat \top$ is true and $K, Q\sat \bot$ is false for any
    set $Q$ in any Kripke structure $K$.
  \item $K, Q\sat a$ iff $a\in L(s)$ for some $s\in Q$, $a\in \ap$.
  \item $K, Q\sat\lnot f$ iff $\lnot (K,Q\sat f)$.
  \item $K, Q\sat f\land g$ iff $K, Q\sat f$ and $K, Q\sat g$.
  \item $K, Q\sat f\lor g$ iff $K, Q\sat f$ or $K, Q\sat g$.
  \item $K, Q\sat \somepath\ f$ for some path formula $f$ iff for some
    $s\in Q$ there exists a path $\pi = s \rightarrow s_{1}
    \rightarrow s_{2}\rightarrow \cdots \rightarrow s_{i}$ such that
    $K, \pi \sat f$.
  \item $K, Q\sat \allpaths\ f$ for some path formula $f$ iff for some
    $s\in Q$ it holds that $K, \pi \sat f$ for all paths $\pi = s
    \rightarrow s_{1} \rightarrow s_{2}\rightarrow \cdots \rightarrow
    s_{i}$.
  \end{enumerate}
\end{definition}

With these definition we can introduce the following equivalence
relation between Kripke structures and LTS.

\begin{definition}
  \label{def-kripke-lts-eq}
  \textsc{Equivalence between Kripke structures and LTS:} Given a
  Kripke structure $K$ and a set of states $Q$ of $K$, the pair $K,Q$
  is equivalent to a process $p$, written $K,Q\peq p$ (or $p\peq
  K,Q$), if and only if for any CTL* formula $f$ $K,Q\sat f$ if and
  only if $p \sat f$.
\end{definition}

It is easy to see that the relation $\peq$ is indeed an equivalence
relation.  This equivalence has the useful property that given an LTS
it is easy to construct its equivalent Kripke structure.

\begin{theorem}
  \label{th-kripke-lts}
  There exists an algorithmic function $\ltstokripke$ which converts a
  labeled transition system $p$ into a Kripke structure $K$ and a set
  of states $Q$ such that $p \peq (K, Q)$.

  Specifically, for any labeled transition system $p = (S,
  A,\rightarrow, s_{0})$, its equivalent Kripke structure $K =
  \ltstokripke(p)$ is defined as $K = (S', Q, R', L')$ where:
  \begin{enumerate}
  \item $S' = \{\langle s, x\rangle : s\in S, x \subseteq \init(s)\}$.
  \item $Q = \{\langle s_{0}, x\rangle \in S'\}$.
  \item $R'$ contains exactly all the transitions $(\langle s,
    N\rangle , \langle t, O\rangle)$ such that $\langle s, N\rangle ,
    \langle t, O\rangle\in S'$, and
    \begin{enumerate}
    \item for any $n\in N$, $s\Suparrow{n}t$,
    \item for some $q\in S$ and for any $o\in O$, $t\Suparrow{o}q$,
      and
    \item \label{th-kripke-lts-3} if $N=\emptyset$ then $O=\emptyset$
      and $t=s$ (these loops ensure that the relation $R'$ is
      complete).
    \end{enumerate}
  \item $L': S'\rightarrow 2^{\ap}$ such that $L'(s, x) = x$, where
    $\ap = A$.
  \end{enumerate}
\end{theorem}

\begin{figure}[tbp]
  \begin{center}
    \begin{tabular}{ccc}
      \mbox{\begin{picture}(28,24)(-4,-2)
          \gasset{Nframe=n,ExtNL=y,NLdist=1,NLangle=0,Nw=2,Nh=2,Nframe=n,Nfill=y}
          \node(p)(10,20){$p$} 
          \node(t)(15,10){$t$}
          \node[NLangle=180](q)(5,10){$q$} \node(s)(10,0){$s$}
          \node[NLangle=180](r)(0,0){$r$} \node(u)(20,0){$u$}
          \drawtrans[r](q,r){$c$} \drawtrans(q,s){$d$}
          \drawtrans[r](p,q){$a$} \drawtrans(p,t){$b$}
          \drawtrans(t,u){$e$}
        \end{picture}}
      & &
      \mbox{\begin{picture}(45,26)(-10,-4)
          \gasset{Nadjust=wh,Nadjustdist=1,Nfill=n}
          \node(p1)(5,18){$p,\{a\}$} 
          \node(q2)(10,9){$q,\{d\}$}
          \node(q1)(0,9){$q,\{c\}$} 
          \node(s)(10,0){$s,\emptyset$}
          \node(r)(0,0){$r,\emptyset$} 
          \node(p2)(25,18){$p,\{b\}$}
          \node(t)(25,9){$t,\{e\}$} 
          \node(u)(25,0){$u,\emptyset$}
          \drawtrans(p1,q1){} 
          \drawtrans(p1,q2){} 
          \drawtrans(q1,r){}
          \drawtrans(q2,s){} 
          \drawtrans(p2,t){} 
          \drawtrans(t,u){}
          \drawloop[l](r){} 
          \drawloop[r](s){} 
          \drawloop[r](u){}
        \end{picture}}
      \\
      $(a)$ & & $(b)$
    \end{tabular}
  \end{center}
  \begin{quote}
    \small
    Each state of the Kripke structure $(b)$ is labeled with the LTS
    state it came from, and the set of propositions that hold in that
    Kripke state.
  \end{quote}
  \caption{A conversion of an LTS $(a)$ to an equivalent Kripke
    structure $(b)$.}
  \label{fig-lts-kripke}
\end{figure}

Figure~\ref{fig-lts-kripke} illustrates graphically how we convert a
labelled transition system to its equivalent Kripke structure.  As
illustrated in the figure, we combine each state in the labelled
transition system with its actions provided as properties to form new
states in the equivalent Kripke structure. The transition relation of
the Kripke structure is formed by the new states and the corresponding
transition relation in the original labelled transition system.  The
labeling function in the equivalent Kripke structure links the actions
to their relevant states.  An LTS states in split into multiple Kripke
states whenever it can evolve differently by performing different
actions (like state $p$ in the figure).

Using such a conversion we can define the semantics of CTL* formulae
with respect to a process rather than Kripke structure.  One
problem---that required the new satisfaction operator for sets of
Kripke states as defined in Definition~\ref{def-kripke-set-sat}---is
introduced by the fact that one state of a process can generate
multiple initial Kripke states.  We believe that the weaker
satisfaction operator from Definition~\ref{def-kripke-set-sat} is
introduced without loss of generality and may even be worked around by
such mechanisms as considering processes with one outgoing transition
(a ``start'' action) followed by their normal behaviour.

\medskip

\begin{proof*}{Theorem~\ref{th-kripke-lts}} The proof relies on the
  properties of the syntax and semantics of CTL* formulae and is done
  by structural induction.

  For the basis of the induction, we note that $\top$ is true for any
  process and for any state in any Kripke structure.  $p\sat \top$ iff
  $K,Q\sat \top$ is therefore immediate.  The same goes for $\bot$ (no
  process and no state in any Kripke structure satisfy $\bot$).
  $p\sat a$ iff $\ltstokripke(p)=K,Q\sat a$ by the definition of
  $\ltstokripke$; indeed, $a \in \init(p)$ (so that $p \sat a$) iff
  $a\in x$ for some state $\langle s,x \rangle\in Q$ that is, $a \in
  L'(\langle s,x\rangle)$.

  On to the inductive step.  $p \sat f'$ iff $\ltstokripke(p) \sat f'$ for any
  formula $f'$ by induction hypothesis, so we take $f'=\lnot f$ and so
  $p \sat \lnot f$ iff $\ltstokripke(p) \sat \neg f$.

  Suppose that $p\sat f$ and [or] $p\sat g$ (so that $p\sat f\land g$
  [$p\sat f\lor g$]).  This is equivalent by induction hypothesis to
  $\ltstokripke(p)\sat f$ and [or] $\ltstokripke(p)\sat g$, that is,
  $\ltstokripke(p)\sat f\land g$ [$\ltstokripke(p)\sat f\lor g$], as
  desired.

  Let now $\pi$ be a path $\pi=p\suparrow{a_{0}} s_{1}\suparrow{a_{1}}
  s_{2}\suparrow{a_{2}} \cdots \suparrow{a_{n}}s_{n}$ starting from a
  process $p$.  According to the definition of $\ltstokripke$, all the
  equivalent paths in the Kripke structure $\ltstokripke(p)$ have the
  form $\pi' = \langle p, A_{0}\rangle \rightarrow \langle s_{1},
  A_{1}\rangle \rightarrow \langle s_{2},
  A_{2}\rangle\rightarrow\cdots \rightarrow \langle s_{n},A_{n}
  \rangle$, such that $a_{i}\in A_{i}$ for all $0\leq i<n$.  Clearly,
  such a path $\pi'$ exists.  Moreover, given some path of form
  $\pi'$, a path of form $\pi$ also exists (because no path in the
  Kripke structure comes out of the blue; instead all of them come
  from paths in the original process).  By abuse of notation we write
  $\ltstokripke(\pi)=\pi'$, with the understanding that this
  incarnation of $\ltstokripke$ is not necessarily a function (it
  could be a relation) but is complete (there exists a path
  $\ltstokripke(\pi)$ for any path $\pi$).  With this notion of
  equivalent paths we can now proceed to path formulae.

  Consider the formula $\xnext f$ such that some path $\pi$ satisfies
  it.  Whenever $\pi \sat \xnext f$, $\pi^{1}\sat f$ and therefore
  $\ltstokripke(\pi)^{1}\sat f$ (by inductive assumption, for indeed
  $f$ is a state, not a path formula) and therefore
  $\ltstokripke(\pi)\sat \xnext f$, as desired.  Conversely,
  $\ltstokripke(\pi)\sat \xnext f$, that is, $\ltstokripke(\pi)^{1}\sat
  f$ means that $\pi^{1}\sat f$ by inductive assumption, and so
  $\pi\sat \xnext f$.
  
  The proof for \eventually, \globally, \until, and \releases\
  operators proceed similarly.  Whenever $\pi\sat \eventually\ f$,
  there is a state $\pi^{i}$ such that $\pi^{i}\sat f$.  By induction
  hypothesis then $\ltstokripke(\pi)^{i}\sat f$ and so
  $\ltstokripke(\pi)\sat \eventually\ f$.  The other way (from
  $\ltstokripke(\pi)$ to $\pi$) is similar.  The \globally\ operator
  requires that all the states along $\pi$ satisfy $f$, which imply
  that all the states in any $\ltstokripke(\pi)$ satisfy $f$, and thus
  $\ltstokripke(\pi)\sat \globally\ f$ (and again things proceed
  similarly in the other direction).  In all, the induction hypothesis
  established a bijection between the states in $\pi$ and the states
  in (any) $\ltstokripke(\pi)$.  This bijection is used in the proof
  for \until\ and \releases\ just as it was used in the above proof
  for \eventually\ and \globally.  Indeed, the states along the path
  $\pi$ will satisfy $f$ or $g$ as appropriate for the respective
  operator, but this translates in the same set of states satisfying
  $f$ and $g$ in $\ltstokripke(\pi)$, so the whole formula (using
  \until\ or \releases\ holds in $\pi$ iff it holds in
  $\ltstokripke(\pi)$.

  Finally, given a formula $\somepath\ f$, $p\sat \somepath\ f$
  implies that there exists a path $\pi$ starting from $p$ that
  satisfies $f$.  By induction hypothesis there is then a path
  $\ltstokripke(\pi)$ starting from $\ltstokripke(p)$ that satisfies
  $f$ (there is at least one such a path) and thus
  $\ltstokripke(p)\sat \somepath\ f$.  The other way around is
  similar, and so is the proof for $\allpaths\ f$ (all the paths $\pi$
  satisfy $f$ so all the path $\ltstokripke(\pi)$ satisfy $f$ as well;
  there are no supplementary paths, since all the paths in
  $\ltstokripke(p)$ come from the paths in $p$).
\end{proof*}

\subsection{Yet Another Constructive Equivalence between LTS and Kripke Structures}
\label{sect-lts-kripke-b}

The function $\ltstokripke$ developed earlier produces a very compact
Kripke structure.  However, a state in the original LTS can result in
multiple equivalent state in the resulting Kripke structure, which in
turn requires a modified notion of satisfaction (over sets of states,
see Definition~\ref{def-kripke-set-sat}).  This in turn implies a
non-standard model checking algorithm.  A different such a conversion
algorithm \cite{nicola90} avoids this issue, at the expense of a
considerably larger Kripke structure.  We now explore a similar
equivalence.

The just mentioned conversion algorithm \cite{nicola90} is based on
introducing intermediate states in the resulting Kripke structure.
These states are labelled with the special proposition $\Delta$ which
is understood to mark a state that is ignored in the process of
determining the truth value of a CTL formula; if $\Delta$ labels a
state then it is the only label for that state.  We therefore base our
construction on the following definition of equivalence between
processes and Kripke structures:

\begin{definition}
  \label{def-Kripke-Kripke1-sat}
  \textsc{Satisfaction for processes:} 
  Given a Kripke structure $K$ and a state $s$ of $K$, the pair $K,s$
  is equivalent to a process $p$, written as $K,s\peq p$ (or
  $p\peq K,s$) iff for any CTL* formula $f$ $K,s\models f$ iff
  $p \models f$.  The operator $\models$ is defined for processes in
  Definition~\ref{def-lts-sat} and for Kripke structures as follows:
  \begin{enumerate}
  \item $p\models \top$  iff $K, s\models \top$
  \item $p\models a$ iff $K, s\models \Delta\ \until\ a$
  \item $p\models\lnot f$ iff $K, s\models\lnot f$
  \item $p\models f\land g$ iff $K, s\models f\land g$
  \item $p\models f\lor g$ iff $K, s\models f\lor g$ 
  \item $p\models\somepath f$ iff $K, s\models\somepath f$
  \item $p\models f\ \until\ g$ iff
    $K, s\models (\Delta\lor\ f)\ \until\ g$
  \item $p\models \xnext\ f$ iff $K, s\models\xnext\ (\Delta\ \until\ f)$
  \item $p\models \eventually f$ iff $K, s\models\eventually\ f$
  \item $p\models \globally\ f$ iff
    $K, s\models\globally\ (\Delta\ \lor\ f)$
  \item $p\models f\ \releases\ g$ iff $K, s\models\ f\ \releases\ (\Delta \lor\ g)$
  \end{enumerate}
\end{definition}

Note that the definition above is stated in terms of CTL* rather than
CTL; however, CTL* is stronger and so equivalence under CTL* implies
equivalence under CTL.

Most of the equivalence is immediate.  However, some cases need to
make sure that the states labelled $\Delta$ are ignored.  This happens
first in $K,s\models \Delta\ \until\ a$, which is equivalent to
$p\models a$.  Indeed, $a$ needs to hold immediately, except that any
preceding states labelled $\Delta$ must be ignored, hence $a$ must be
eventually true and when it becomes so it releases the chain of
$\Delta$ labels.  The formula for $\xnext$ is constructed using the
same idea (except that the formula $f$ releasing the possible chain of
$\Delta$ happens starting from the next state).

Then expression $(\Delta\lor\ f)\ \until\ g$ means that $f$ must
remain true with possible interleaves of $\Delta$ until $g$ becomes
true.  Similarly $f\ \releases\ (\Delta \lor\ g)$ requires that $g$ is
true (with the usual interleaved $\Delta$) until it is released by $f$
becoming true.

Based on this equivalence we can define a new conversion of LTS into
equivalent Kripke structures.  This conversion is again based on a
similar conversion \cite{nicola90} developed in a different context.

\begin{theorem} 
  \label{new-th-kripke-lts}
  There exist at least two algorithmic functions for converting LTS
  into equivalent Kripke structures.  The first is the function
  $\ltstokripke$ described in Theorem~\ref{th-kripke-lts}.

  The new function $\ltstokripkex$ is defined as follows: with
  $\Delta$ a fresh symbol not in $A$: Given an LTS
  $p = (S, A, \rightarrow, s_{0})$, the Kripke structure
  $\ltstokripkex(p) = (S', Q, R', L)$ is given by:
  \begin{enumerate}
  \item $\ap = A\uplus{\Delta}$;
  \item $S' = S\cup\{(r,a,s): a\in A\ \mbox{ and } r\suparrow{a}s\}$;
  \item $Q = \{s_{0}\}$;
  \item $R'=\{(r,s): r\suparrow{\tau} s\}\cup \{(r,(r,a,s)): r\suparrow{a}s\}\cup\{((r,a,s),s): r\suparrow{a}s\}$; 
  \item For $r,s\in S\ and\ a\in A:L(s)=\{\Delta\}\ and\ L((r,a,s))=\{a\}$.
  \end{enumerate}
  Then $p \peq \ltstokripkex(p)$.
\end{theorem}

\begin{proof}
  We prove the stronger equivalence over CTL* rather than CTL by
  structural induction.  Since $\Delta$ is effectively handled by the
  satisfaction operator introduced in
  Definition~\ref{def-Kripke-Kripke1-sat} it will turn out that there
  is no need to mention it at all.

  For the basis of the induction, we note that $\top$ is true for any
  process and for any state in any Kripke structure.  $p\models \top$
  iff $\ltstokripkex(p)\models \top$ is therefore immediate.  The same
  goes for $\bot$ (no process and no state in any Kripke structure
  satisfy $\bot$).  $p\models a$ iff $\ltstokripkex(p)\models a$;
  Indeed, $a \in\ A$(so that $p\models a$) iff $a \in\ L((r,a,s))$.

  That $p \models \lnot f$ iff $\ltstokripkex(p) \models \neg f$ is
  immediately given by the induction hypothesis that $p \models f$ iff
  $\ltstokripkex(p) \models f$.

  Suppose that $p\models f$ and [or] $p\models g$ (so that
  $p\models f\land g$ [$p\models f\lor g$]).  This is equivalent by
  induction hypothesis to $\ltstokripkex(p)\models f$ and [or]
  $\ltstokripkex(p)\models g$, that is,
  $\ltstokripkex(p)\models f\land g$
  [$\ltstokripkex(p)\models f\lor g$], as desired.

  Let now $\pi'$ be a path
  $\pi'=p\suparrow{a_{0}} s_{1}\suparrow{a_{1}} s_{2}\suparrow{a_{2}}
  \cdots \suparrow{a_{n}}s_{n}$ starting from a process $p$.
  According to the definition of $\ltstokripkex$, all the equivalent
  paths in the Kripke structure $\ltstokripkex(p)$ have the form
  $\pi' =\Delta \rightarrow A_{0} \rightarrow \Delta \rightarrow
  A_{1}\rightarrow \Delta \rightarrow A_{2}\rightarrow\cdots \Delta
  \rightarrow A_{n}$, such that $a_{i}\in A_{i}$ for all $0\leq
  i<n$. Clearly, such a path $\pi'$ exists. According to the function
  $\ltstokripkex$, we know that $\Delta$ is a symbol that stands for
  states in the LTS and has no meaning in the Kripke structure.  The
  satisfaction operator for Kripke structures
  (Definition~\ref{def-Kripke-Kripke1-sat}) is specifically designed
  to ignore the $\Delta$ label and this insures that the part $\pi'$
  is equivalent to the path
  $\pi=A_{0} \rightarrow A_{1}\rightarrow A_{2}\rightarrow\cdots
  \rightarrow A_{n}$ with $a_{i}\in A_{i}$ for all $0\leq i<n$ and so
  we will use this form for the reminder of the proof.
  
  Consider the formula $\xnext f$ such that some path $\pi$ satisfies
  it.  Whenever $\pi \models \xnext f$, $\pi^{1}\models f$ and
  therefore $\ltstokripkex(\pi)^{1}\models f$ (by inductive
  assumption, for indeed $f$ is a state, not a path formula) and
  therefore $\ltstokripkex(\pi)\models \xnext f$, as desired.
  Conversely, $\ltstokripkex(\pi)\models \xnext f$, that is,
  $\ltstokripkex(\pi)^{1}\models f$ means that $\pi^{1}\models f$ by
  inductive assumption, and so $\pi\models \xnext f$.
  
  The proof for \eventually, \globally, \until, and \releases\
  operators proceed similarly.  Whenever $\pi\models \eventually\ f$,
  there is a state $\pi^{i}$ such that $\pi^{i}\models f$.  By
  induction hypothesis then $\ltstokripkex(\pi)^{i}\models f$ and so
  $\ltstokripkex(\pi)\models \eventually\ f$.  The other way (from
  $\ltstokripkex(\pi)$ to $\pi$) is similar.  The \globally\ operator
  requires that all the states along $\pi$ satisfy $f$, which implies
  that all the states in any $\ltstokripkex(\pi)$ satisfy $f$, and
  thus $\ltstokripkex(\pi)\models \globally\ f$ (and again things
  proceed similarly in the other direction).  In all, the induction
  hypothesis established a bijection between the states in $\pi$ and
  the states in (any) $\ltstokripkex(\pi)$.  This bijection is used in
  the proof for \until\ and \releases\ just as it was used in the
  above proof for \eventually\ and \globally.  Indeed, the states
  along the path $\pi$ will satisfy $f$ or $g$ as appropriate for the
  respective operator, but this translates in the same set of states
  satisfying $f$ and $g$ in $\ltstokripkex(\pi)$, so the whole formula
  (using \until\ or \releases\ holds in $\pi$ iff it holds in
  $\ltstokripkex(\pi)$).

  Finally, given a formula $\somepath\ f$, $p\models \somepath\ f$
  implies that there exists a path $\pi$ starting from $p$ that
  satisfies $f$.  By induction hypothesis there is then a path
  $\ltstokripkex(\pi)$ starting from $\ltstokripkex(p)$ that satisfies
  $f$ (there is at least one such a path) and thus
  $\ltstokripkex(p)\models \somepath\ f$.  The other way around is
  similar, and so is the proof for $\allpaths\ f$ (all the paths $\pi$
  satisfy $f$ so all the path $\ltstokripkex(\pi)$ satisfy $f$ as well;
  there are no supplementary paths, since all the paths in
  $\ltstokripkex(p)$ come from the paths in $p$).
\end{proof}

\begin{figure}[tbp]
  \centering
  \begin{picture}(28,38)(-4,-3)
    \gasset{Nframe=n,ExtNL=y,NLdist=1,NLangle=0,Nw=2,Nh=2,Nframe=n,Nfill=y}
    \node(p)(10,32){$p$}
    \node[NLangle=180](q)(5,16){$q$}
    \node(t)(15,16){$t$}
    \node[NLangle=180](r)(0,0){$r$}
    \node(s)(10,0){$s$}
    \node(u)(20,0){$u$}

    \drawtrans[r](p,q){$a$}
    \drawtrans(p,t){$b$}
    \drawtrans[r](q,r){$c$}
    \drawtrans(q,s){$d$}
    \drawtrans(t,u){$e$}
  \end{picture}
  \qquad
  \begin{picture}(26,38)(-3,-3)
    \gasset{Nadjust=wh,Nadjustdist=1,Nfill=n}
    \node(t1)(12.5,32){$\triangle$}
    \node(a)(5,24){$a$}
    \node(b)(20,24){$b$}
    \node(t2)(5,16){$\triangle$}
    \node(t3)(20,16){$\triangle$}
    \node(c)(0,8){$c$}
    \node(d)(10,8){$d$}
    \node(e)(20,8){$e$}
    \node(t4)(0,0){$\triangle$}
    \node(t5)(10,0){$\triangle$}
    \node(t6)(20,0){$\triangle$}

    \drawtrans(t1,a){}
    \drawtrans(t1,b){}
    \drawtrans(a,t2){}
    \drawtrans(b,t3){}
    \drawtrans(t2,c){}
    \drawtrans(t2,d){}
    \drawtrans(t3,e){}
    \drawtrans(c,t4){}
    \drawtrans(d,t5){}
    \drawtrans(e,t6){}
  \end{picture}
  \\
  $(a)$ \qquad\qquad\qquad\qquad\qquad\qquad $(b)$
  \caption{Conversion of an LTS $(a)$ to its equivalent Kripke structure $(b)$.}
  \label{new-fig-lts-ks}
\end{figure}
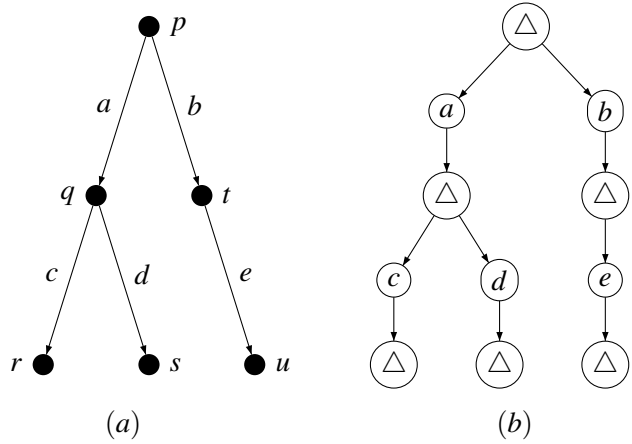

The process of the new version described in
Theorem~\ref{new-th-kripke-lts} is most easily described graphically;
refer for this purpose to Figure~\ref{new-fig-lts-ks}.  Specifically,
the function $\ltstokripkex$ converts the LTS given in Figure
~\ref{new-fig-lts-ks}$(a)$ into the equivalent Kripke structure shown
in Figure~\ref{new-fig-lts-ks}$(b)$.  In this new structure, instead
of combining each state with its corresponding actions in the LTS (and
thus possibly splitting the LTS state into multiple Kripke structure
states), we use the new symbol $\Delta$ to stand for the original LTS
states.  Every $\Delta$ state of the Kripke structure is the LTS
state, and all the other states in the Kripke structure are the
actions in the LTS.  This ensures that all states in the Kripke
structure corresponding to actions that are outgoing from a single LTS
state have all the same parent.  This in turn eliminates the need for
the weaker satisfaction operator over sets of states
(Definition~\ref{def-kripke-set-sat}).

\section{CTL Is Equivalent to Failure Trace Testing}
\label{sect-eq-ctl-ft}

We now proceed to show the equivalence between CTL formulae and
failure trace tests.  Let \procset\ be the set of all processes,
\ftset\ the set of all failure trace tests, and \ctlset\ the set of
all CTL formulae.  We have:

\begin{theorem}
  \label{th-ctl-ft-main}
  \begin{enumerate}
  \item For some $t\in \ftset$ and $f\in \ctlset$, whenever
    $p\ \may\ t$ if and only if $\ltstokripke(p) \sat f$ for any
    $p\in \procset$ we say that $t$ and $f$ are equivalent.  Then, for
    every failure trace test there exists an equivalent CTL formula
    and the other way around.  Furthermore a failure trace test can be
    algorithmically converted into its equivalent CTL formula and the
    other way around.
  \item For some $t\in \ftset$ and $f\in \ctlset$, whenever
    $p\ \may\ t$ if and only if $\ltstokripkex(p) \sat f$ for any
    $p\in \procset$ we say that $t$ and $f$ are equivalent.  Then, for
    every failure trace test there exists an equivalent CTL formula
    and the other way around.  Furthermore a failure trace test can be
    algorithmically converted into its equivalent CTL formula and the
    other way around.
\end{enumerate}
\end{theorem}

\begin{proof}
  The proof of Item~1 follows from Lemma~\ref{th-ft-ctl} (in
  Section~\ref{sect-ft-to-ctl} below) and Lemma~\ref{th-ctl-ft} (in
  Section~\ref{sect-ctl-to-ft} below).  The algorithmic nature of the
  conversion is shown implicitly in the proofs of these two results.
  The proof of Item~2 is fairly similar and is summarized in
  Lemma~\ref{th-ft-ctl-x} (Section~\ref{sect-ft-to-ctl}) and
  Lemma~\ref{th-ctl-ft-x} (Section~\ref{sect-ctl-to-ft}).
\end{proof}

The remainder of this section is dedicated to the proof of the lemmata
mentioned above and so the actual proof of this result.  Note
incidentally that Lemmata~\ref{th-ctl-ft} and~\ref{th-ctl-ft-x} will
be further improved in Theorem~\ref{th-ctl-ftr-compt}.

\subsection{From Failure Trace Tests to CTL Formulae}
\label{sect-ft-to-ctl}

\begin{lemma}
  \label{th-ft-ctl}
  There exists a function $\fttoctl_{\ltstokripke}:\ftset\rightarrow\ctlset$ such
  that $p\ \may\ t$ if and only if $\ltstokripke(p)\sat \fttoctl_{\ltstokripke}(t)$ for any
  $p\in\procset$.
\end{lemma}

\begin{proof}
  The proof is done by structural induction over tests.  In the
  process we also construct (inductively) the function $\fttoctl_{\ltstokripke}$.

  We put $\fttoctl_{\ltstokripke}(\pass) = \top$. Any process passes \pass\ and any
  Kripke structure satisfies $\top$, thus it is immediate that $p\
  \may\ \pass$ iff $\ltstokripke(p)\sat\top = \fttoctl_{\ltstokripke}(\pass)$.  Similarly, we put
  $\fttoctl_{\ltstokripke}(\pstop) = \bot$. No process passes \pstop\ and no Kripke
  structure satisfies $\bot$.  

  On to the induction steps now.  We put $\fttoctl_{\ltstokripke}(\mathbf{i}; t) =
  \fttoctl_{\ltstokripke}(t)$: an internal action in a test is not seen by the process
  under test by definition.  We then put $\fttoctl_{\ltstokripke}(a; t) = a \land
  \somepath\xnext\ \fttoctl_{\ltstokripke}(t)$.  We note that $p\ \may\ (a; t)$ iff $p\
  \may\ a$ and $p'\ \may\ t$ for some $p\suparrow{a}p'$.  Now, $p\
  \may\ a$ iff $\ltstokripke(p)\sat a$ by the construction of $\ltstokripke$, and also
  $p'\ \may\ t$ iff $\ltstokripke(p')\sat \fttoctl_{\ltstokripke}(t)$ by induction hypothesis.  By
  Theorem~\ref{th-kripke-lts}, when we convert $p$ to an equivalent
  Kripke structure $\ltstokripke(p)$ we take as new states the original states
  together with their outgoing actions. So once we are (in $\ltstokripke(p)$)
  in a state that satisfies $a$, all the next states of that state
  correspond to the states following $p$ after executing $a$.
  Therefore, $\xnext(\fttoctl_{\ltstokripke}(t))$ is satisfied in exactly those states in
  which $t$ must succeed.  Thus $p\ \may\ a;t$ iff $\ltstokripke(p)\sat a\land
  \somepath\xnext \fttoctl_{\ltstokripke}(t)$.  For illustration purposes note that in
  Figure~\ref{fig-lts-kripke} the initial state $p$ becomes two
  initial states $(p, \{a\})$ and $(p, \{b\})$; the next state of the
  state satisfying the property $a$ in the Kripke structure contains
  only $q$ (and never $t$).

  Note now that $\choice$ is just syntactical sugar, for indeed
  $t_{1}\ \choice\ t_{2}$ is perfectly equivalent with
  $\mchoice\{t_{1},t_{2}\}$.  We put\footnote{As usual $\Lor \{t_{1},
    \ldots, t_{n}\}$ is a shorthand for $t_{1}\lor \cdots\lor t_{n}$.}
  $\fttoctl_{\ltstokripke}(\mchoice T) = \Lor \{ \fttoctl_{\ltstokripke}(t): t\in T\}$.  $p\ \may\
  \mchoice T$ iff $p\ \may\ t$ for at least one $t\in T$ iff
  $\ltstokripke(p)\sat \fttoctl_{\ltstokripke}(t)$ for at least one $t\in T$ (by
  induction hypothesis) iff $\ltstokripke(p)\sat \Lor \{ \fttoctl_{\ltstokripke}(t):
  t\in T\}$.

  We finally get to consider $\theta$.  Note first that whenever
  $\theta$ does not participate in a choice it behaves exactly like
  $\mathbf{i}$, so we assume without loss of generality that $\theta$
  appears only in choice constructs.  We also assume without loss of
  generality that every choice contains at most one top-level
  $\theta$, for indeed $\theta;t_{1}\ \choice\ \theta;t_{2}$ is
  equivalent with $\theta;(t_{1}\ \choice\ t_{2})$.  We put
  $\fttoctl_{\ltstokripke}(t_{1}\ \choice\ \theta;t) = ((\Lor \init(t_{1})) \land
  \fttoctl_{\ltstokripke}(t_{1})) \lor (\lnot (\Lor \init(t_{1})) \land
  \fttoctl_{\ltstokripke}(t))$.  

  According to the TLOTOS definition of $\parallelth$ (see
  Section~\ref{sect-ft}), if a common action is available for both $p$
  and $t$ then the deadlock detection action $\theta$ will not play
  any role.  In other words, whenever $p\Suparrow{a}$ such that $a\in
  \init(t_{1})$ we have $p\ \may\ t_{1}\ \choice\ \theta;t$ iff $p\
  \may\ t_{1}$.  We further note that $p\Suparrow{a}$ is equivalent to
  $\ltstokripke(p)\sat a$ and so given the inductive hypothesis (that
  $p'\ \may\ t_{1}$ iff $\ltstokripke(p')\sat \fttoctl_{\ltstokripke}(t_{1})$ for any
  process $p'$) we conclude that:
  \begin{equation}
    \label{eq-th-ft-ctl-1}
    (p\ \may\ t_{1}\ \choice\ \theta;t)\ \land\ (p\Suparrow{a}
    \land a\in \init(t_{1})) \mbox{\quad iff \quad} \ltstokripke(p)\sat (\Lor
    \init(t_{1})) \land \fttoctl_{\ltstokripke}(t_{1}) 
  \end{equation}
  Whenever it is not the case that $p\Suparrow{a}$ and $a\in
  \init(t_{1})$ (equivalent to $\ltstokripke(p)\not\sat\Lor \init(t_{1})$),
  then the deadlock detection transition $\theta$ of $t_{1}\ \choice\
  \theta;t$ will fire and then the test will succeed iff $t$ succeeds.
  Given once more the inductive hypothesis that $p'\ \may\ t$ iff
  $\ltstokripke(p')\sat \fttoctl_{\ltstokripke}(t)$ for any process $p'$ we have:
  \begin{equation}
    \label{eq-th-ft-ctl-2}
    (p\ \may\ t_{1}\ \choice\ \theta;t)\ \land\ \lnot (p\Suparrow{a}
    \land a\in \init(t_{1})) \mbox{\quad iff \quad} \ltstokripke(p)\sat \lnot (\Lor
    \init(t_{1})) \land \fttoctl_{\ltstokripke}(t) 
  \end{equation}
  The correctness of $\fttoctl_{\ltstokripke}(t_{1}\ \choice\ \theta;t)$ is then a
  direct consequence of Relations~(\ref{eq-th-ft-ctl-1})
  and~(\ref{eq-th-ft-ctl-2}); indeed, one just has to take the
  disjunction of both sides of these relations to reach the desired
  equivalence:
  \[
  p\ \may\ t_{1}\ \choice\ \theta;t \mbox{\quad iff \quad}
  \ltstokripke(p)\sat (\Lor \init(t_{1})) \land \fttoctl_{\ltstokripke}(t_{1})\ \lor\
  \ltstokripke(p)\sat \lnot (\Lor \init(t_{1})) \land \fttoctl_{\ltstokripke}(t)
  \]
  The induction is thus complete.
\end{proof}

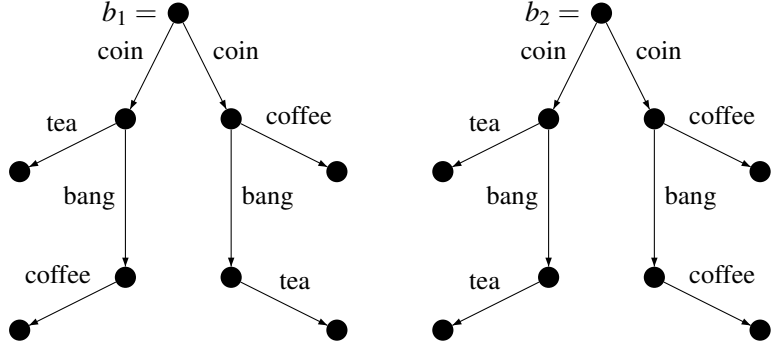
\begin{figure}[tbp]
  \centering
  \begin{tabular}{ccc}
    \mbox{\begin{picture}(34,34)(-2,-2)
        \gasset{Nframe=n,ExtNL=y,NLdist=1,NLangle=180,Nw=2,Nh=2,Nframe=n,Nfill=y}
        
        \node(b)(15,30){$b_{1}=$}
        \node(u1)(10,20){}
        \node(u2)(20,20){}
        \node(m1)(0,15){}
        \node(m2)(30,15){}
        \node(l2)(20,5){}
        \node(l1)(10,5){}
        \node(b2)(30,0){}
        \node(b1)(0,0){}

        \drawtrans[r](l1,b1){\small coffee}
        \drawtrans[l](l2,b2){\small tea}
        \drawtrans[r](u1,l1){\small bang}
        \drawtrans[l](u2,l2){\small bang}
        \drawtrans[r](u1,m1){\small tea}
        \drawtrans[l](u2,m2){\small coffee}
        \drawtrans[r](b,u1){\small coin}
        \drawtrans[l](b,u2){\small coin}
      \end{picture}}
    & \qquad &
        \mbox{\begin{picture}(34,34)(-2,-2)
        \gasset{Nframe=n,ExtNL=y,NLdist=1,NLangle=180,Nw=2,Nh=2,Nframe=n,Nfill=y}

        \node(b)(15,30){$b_{2}=$}
        \node(u1)(10,20){}
        \node(u2)(20,20){}
        \node(m1)(0,15){}
        \node(m2)(30,15){}
        \node(l2)(20,5){}
        \node(l1)(10,5){}
        \node(b2)(30,0){}
        \node(b1)(0,0){}

        \drawtrans[r](l1,b1){\small tea}
        \drawtrans[l](l2,b2){\small coffee}
        \drawtrans[r](u1,l1){\small bang}
        \drawtrans[l](u2,l2){\small bang}
        \drawtrans[r](u1,m1){\small tea}
        \drawtrans[l](u2,m2){\small coffee}
        \drawtrans[r](b,u1){\small coin}
        \drawtrans[l](b,u2){\small coin}
      \end{picture}}
  \end{tabular}
  \caption{Two coffee machines.}
  \label{fig-coffee}
\end{figure}

\begin{qexample}{How to tell logically that your coffee machine is
    working.}
  \label{ex-coffee}
  The coffee machines $b_{1}$ and $b_{2}$ below were famously
  introduced to illustrate the limitations of the may and must testing
  framework of De~Nicola and Hennessy.  Indeed, they have been found
  \cite{langerak89} to be equivalent under testing preorder
  \cite{denicola84} but not equivalent under stable failure preorder
  \cite{langerak89}.
  \begin{eqnarray*}
    b_{1} &=& \coin; (\tea\ \choice\ \bang; \coffee)\quad \choice
    \quad \coin; (\coffee\ \choice\ \bang; \tea)\\
    b_{2} &=& \coin; (\tea\ \choice\ \bang; \tea)\quad \choice
    \quad \coin; (\coffee\ \choice\ \bang; \coffee)
  \end{eqnarray*}
  These machines are also shown graphically (as LTS) in
  Figure~\ref{fig-coffee}.

  The first machine accepts a coin and then dispenses either tea or
  coffee, at its discretion.  Still, if one wants the other beverage,
  one just hits the machine.  The second machine is rather stubborn,
  giving either tea or coffee at its discretion.  By contrast with the
  first machines, the beverage offered will not be changed by hits.
  One failure trace test that differentiate these machines
  \cite{langerak89} is
  \[ t= \coin; (\coffee; \pass\ \choice\ \theta;
  \bang; \coffee; \pass)\]

  The conversion of the failure trace test $t$ into a CTL formula as
  per Lemma~\ref{th-ft-ctl} yields $\fttoctl_{\ltstokripke}(t) = \coin \land
  \somepath \xnext\ (\coffee \land \somepath\xnext\ \top \lor \lnot
  (\coffee \land \somepath\xnext\ \top) \land \bang \land
  \somepath\xnext\ (\coffee \land \somepath\xnext \top))$.
  Eliminating the obviously true sub-formulae, we obtain the logically
  equivalent formula 
  \[
    \fttoctl_{\ltstokripke}(t) = \coin \land \somepath \xnext\ (\coffee \lor \lnot
    \coffee \land \bang \land \somepath\xnext\ \coffee)
  \] 
  The meaning of this formula is clearly equivalent to the meaning of
  $t$, as it reads quite literally ``a coin is expected, and in the
  next state either coffee is offered, or coffee is not offered but a
  bang is available and then the next state will offer coffee.''  This
  is rather dramatically identical to how one would describe the
  behaviour of $t$ in plain English.

  The formula $\fttoctl_{\ltstokripke}(t)$ holds for both the initial states of
  $\ltstokripke(b_{1})$ (where coffee is offered from the outset or
  follows a hit on the machine) but holds in only one of the initial
  states of $\ltstokripke(b_{2})$ (the one that dispenses coffee).
\end{qexample}

\begin{lemma}
  \label{th-ft-ctl-x}
  There exists a function
  $\fttoctl_{\ltstokripkex}:\ftset\rightarrow\ctlset$ such that
  $p\ \may\ t$ if and only if
  $\ltstokripkex(p)\models \fttoctl_{\ltstokripkex}(t)$ for any
  $p\in\procset$.
\end{lemma}

\begin{proof}
  The proof for $\fttoctl_{\ltstokripke}$ (Lemma~\ref{th-ft-ctl}) holds in almost
  all the cases.  Indeed, the way the operator $\models$ is defined
  (Definition~\ref{def-Kripke-Kripke1-sat}) ensures that all
  occurrences of $\Delta$ are ``skipped over'' as if they were not
  there in the first place.  However the way LTS states are split by
  $\ltstokripke$ facilitates the proof or Lemma~\ref{th-ft-ctl}, yet
  such a split no longer happens in $\ltstokripkex$.  The original
  proof had $\fttoctl_{\ltstokripke}(a;t) = a\land\somepath\xnext\ \fttoctl_{\ltstokripke}(t)$ but
  under $\ltstokripkex$ this construction will fail to work correctly
  on LTS such as $p=a\ \choice\ b; p'$ such that $p'\ \may\ t$.
  Indeed, $\ltstokripkex(p)$ features a node labeled $\Delta$ with two
  children; the first child is labeled $a$ while the second child is
  not but has $\ltstokripkex(p')$ as an eventual descendant (through a
  possible chain of nodes labeled $\Delta$).  Clearly it is not the
  case that $p\ \may\ a;t$, yet
  $\ltstokripkex(p)\models a\land\somepath\xnext\
  \fttoctl_{\ltstokripkex}(t)$, which shows that such a simple
  construction is not sufficient for $\ltstokripkex$.

  To remedy this we set
  $\fttoctl_{\ltstokripkex}(a;t) = \somepath\ (a\ \land\
  \allpaths\xnext\ \lnot a)\ \until\ \fttoctl_{\ltstokripkex}(t)
  \land \somepath\xnext\ \fttoctl_{\ltstokripkex}(t)$.  The second
  term in the conjunction ensures that $\fttoctl_{\ltstokripkex}(t)$
  will hold in some next state, while the first term specifies that a
  run of $a$ will be followed by $\fttoctl_{\ltstokripkex}(t)$ (the
  $a\ \until\ \fttoctl_{\ltstokripkex}(t)$ component) and also that
  the run of $a$ is exactly one state long (the
  $\allpaths\xnext\ \lnot a$ part).  Note in passing that the $\until$
  operator is necessary in order to make sure that $a$ and
  $\fttoctl_{\ltstokripkex}(t)$ are on the same path, for otherwise
  the example used above to show that the original $\ltstokripke(p)$
  does not work here will continue to be in effect.

  The rest of the proof remains unchanged from the proof of
  Lemma~\ref{th-ft-ctl}.
\end{proof}

\begin{qexample}{How to tell logically that your coffee machine is
    working, take 2}
  \label{ex-coffee-x}
  Consider again the coffee machines $b_{1}$ and $b_{2}$ from
  Example~\ref{ex-coffee}.  As discussed earlier, one failure trace
  test that differentiate these machines is
  \[ t= \coin; (\coffee; \pass\ \choice\ \theta;
    \bang; \coffee; \pass)\]
 
  The conversion of the failure trace test $t$ in to a CTL formula
  (after eliminating all the trivially true sub-formulae) will be:
  \[
    \begin{array}{lcl}
      \fttoctl_{\ltstokripkex}(t) 
      &=& \somepath(\coin \land \allpaths\xnext \lnot\coin\ \until\ \coffee)\land \somepath\xnext(\coffee)\lor\\
      &&\somepath(\coin \land \allpaths\xnext \lnot\coin\ \until\ \bang)\land\\  
      &&\somepath\xnext((\bang \land \allpaths\xnext\lnot\bang\ \until\ \coffee) \land \somepath\xnext(\coffee))
    \end{array}
  \]
  This formula specifies that after a coin we can get a coffee
  immediately or after a bang we get a coffee immediately, The meaning
  of this formula is clearly equivalent to the meaning of $t$. As
  expected the formula $\fttoctl_{\ltstokripkex}(t)$ holds for
  $b_{1}$ but not for $b_{2}$.
\end{qexample}

\subsection{Converting Failure Trace Tests into Compact CTL Formulae}
\label{sect-cycle}

Whenever all the runs of a test are finite then the conversion shown
in Lemma~\ref{th-ft-ctl} will produce a reasonable CTL formula.  That
formula is however not in its simplest form.  In particular, the
conversion algorithm follows the run of the test step by step, so
whenever the test has one or more cycles (and thus features
potentially infinite runs) the resulting formula has an infinite
length.  We now show that more compact formulae can be obtained and in
particular finite formulae can be derived out of tests with
potentially infinite runs.  This extension works for both
$\fttoctl_{\ltstokripke}$ and $\fttoctl_{\ltstokripkex}$.

\begin{theorem}
  \label{th-ctl-ftr-compt}
  Let
  $\fttoctl\in\{\fttoctl_{\ltstokripke},\fttoctl_{\ltstokripkex}\}$.
  Then there exists an extension of $\fttoctl$ (denoted by abuse of
  notation $\fttoctl$ as well) such that
  $\fttoctl:\ftset\rightarrow\ctlset$, $p\ \may\ t$ if and only if
  $\ltstokripke(p)\models \fttoctl(t)$ for any $p\in\procset$, and
  $\fttoctl(t)$ is finite for any test $t$ provided that we are
  allowed to mark some entry action $a$ so that we can refer to it as
  either $a$ or $\start(a)$ in each loop of $t$.  An entry action for
  a loop is defined as an action labeling an outgoing edge from a
  state that has an incoming edge from outside the loop.
\end{theorem}

\begin{proof}
  It is enough to show how to produce a finite formula starting from a
  general ``loop'' test.  Such a conversion can be then applied to all
  the loops one by one, relying on the original conversion function
  from Lemma~\ref{th-ft-ctl} or Lemma~\ref{th-ft-ctl-x} (depending on
  whether $\fttoctl$ extends $\fttoctl_{\ltstokripke}$ or
  $\fttoctl_{\ltstokripkex}$) for the rest of the test.  Given the
  reliance on the mentioned lemmata we obtain overall an inductive
  construction.  Therefore nested loops in particular will be
  converted inductively (that is, from the innermost loop to the
  outermost loop).

  Thus to complete the proof it is enough to show how to obtain an
  equivalent, finite CTL formula for the following, general form of
  a loop test:
  \[
    t = a_{0}; (t_{0}\ \choice\ a_{1}; (t_{1}\ \choice\ \cdots a_{n-1}; (t_{n-1}\ \choice\ t) \cdots ))
  \]
  The loop itself consists of the actions $a_{0}$, \ldots, $a_{n-1}$.
  Each such an action $a_{i}$ has the ``exit'' test $t_{i}$ as an
  alternative.  We make no assumption about the particular form of
  $t_{i}$, $0\leq i < n$.

  Given the intended use of our function, this proof will be done
  within the inductive assumptions of the proof of
  Lemmata~\ref{th-ft-ctl} and~\ref{th-ft-ctl-x}.  We will therefore
  consider that the formulae $\fttoctl(a_{i})$ and $\fttoctl(t_{i})$
  exist and are finite, $0\leq i< n$.
  
  We have:
  \[\fttoctl(t)= \left(\somepath\left(\Lor_{i=0}^{n-1} C_{i}\right)\ \until\ 
      \left(\Lor_{i=0}^{n-1} E_{i}\right)\right) \lor\fttoctl(t_{0})
  \]
  where $C_{i}$ represents the cycle in its various stages such that
  \[C_{i}=\somepath\globally(\fttoctl(a_{i})\land\somepath\xnext(\fttoctl(a_{(i+1)\bmod
      n})\land\somepath\xnext\cdots\land\somepath\xnext(\fttoctl(a_{(i+n-1)\bmod
      n}))\cdots ))\] and each $E_{i}$ represents one possible exit
  from the cycle and so
  \[E_{i} = \actcount(a_{i})\land\ \somepath\xnext\ \fttoctl(t_{i})\] 
  with
  \[
    \actcount(a_{i})=\somepath\globally\
    \start(a_{j_{0}})\land\somepath\xnext\
    (a_{j_{1}}\land\cdots\somepath\xnext\ a_{j_{i-1}})
  \] 
  where the sequence $(j_{0}, j_{1}, \ldots, j_{i-1})$ is the
  subsequence of $(0, 1, \ldots, i-1)$ that contains exactly all the
  indices $p$ such that $a_{p}\neq \tau$.  It is worth noting that the
  second term in the disjunction ($\fttoctl(t_{0})$) accounts for the
  possibility that while running $t$ we exit immediately upon entering
  the cycle through the exit test $t_{0}$, in effect without
  traversing any portion of the loop.

  Intuitively, each $C_{i}$ corresponds to $a_{i}$ as being available
  in the test loop, followed by all the rest of the loop in the
  correct order.  It therefore models the decision of the test to
  perform $a_{i}$ and remain in the loop.  Whenever some $a_{i}$ is
  available (``true'') then the corresponding $C_{i}$ is true and so
  the disjunction of the formulae $C_{i}$ will keep being true as long
  as we stay in the loop. By contrast, each $E_{i}$ corresponds to
  being in the right place for the test $t_{i}$ to be available (and
  so $\actcount(a_{i})$ being true), combined with the exit from the
  loop using the test $t_{i}$.  The formula $E_{i}$ will become true
  whenever the test is in the right place and the test $t_{i}$
  succeeds.  Such an event releases the loop formula from its
  obligations (following the semantics of the \until\ operator), so
  such a path can be taken by the test and will be successful.  Like
  the name implies, the function of $\actcount(a_{i})$ is like a
  counter for how many actions separate the test $t_{i}$ from the
  marked action $a_{0}$.  By counting actions we know what test is
  available to exit from the loop (depending on how many actions away
  we are from the start of the loop).

  The formula above assumes that neither the actions in the cycle nor
  the top-level actions of the exit tests are $\theta$.  We introduce
  the deadlock detection action along the following cases, with $k$ an
  arbitrary value, $0\leq k < n$: $\theta$ may appear in the loop as
  $a_{k}$ but not on top level of the alternate exit test
  $t_{k-1 \bmod n}$ (Case 1), on the top level of the test
  $t_{k-1 \bmod n}$ but not as alternate $a_{k}$ (Case 2), or both as
  $a_{k}$ and on the top level of the alternate $t_{k-1 \bmod n}$
  (Case 3).  Given that $\theta$ only affects the top level of the
  choice in which it participates, these cases are exhaustive.
   
  \begin{enumerate}
  \item If any $a_{k} =\theta$ and
    $\theta\not\in \init(t_{k-1 \bmod n})$ then we replace all
    occurrences of $\fttoctl(a_{k})$ in $\fttoctl(t)$ with
    $\lnot(\Lor_{\scriptsize b\in \init (t_{k-1 \bmod n})}
    \fttoctl(b))$ in conjunction with
    $\Lor_{\scriptsize b\in \init(t_{k})\setminus\{\theta\}}
    \fttoctl(b)$ for the ``exit'' formulae and with
    $\fttoctl(a_{k+1 \bmod n})$ for the ``cycle'' formulae).
    Therefore
    $C_{i} = \somepath\globally\
    (\fttoctl(a_{i})\land\somepath\xnext\ (\cdots
    \land\somepath\xnext\ (\lnot(\Lor_{\scriptsize b\in \init (t_{k-1
        \bmod n})} \fttoctl(b))\land \fttoctl(a_{k+1 \bmod n})
    \land\somepath\xnext\ \cdots\ \land\somepath\xnext
    (\fttoctl(a_{(i+n-1)\bmod n}))\cdots)))$ and
    $E_{k} = \actcount(a_{k-1 \bmod n})\land\lnot(\Lor_{\scriptsize
      b\in \init (t_{k-1 \bmod n})} \fttoctl(b))\land
    \Lor_{\scriptsize b\in \init(t_{k})\setminus\{\theta\}}
    \fttoctl(b)\land\somepath\xnext(\fttoctl(t_{k}))$.

  \item If $\theta\in \init(t_{k-1 \bmod n})$ and $a_{k}\neq \theta$
    then we change the exit formula $E_{k-1 \bmod n}$ so that it
    contains two components.  If any action in
    $\init(t_{k-1 \bmod n})$ is available then such an action can be
    taken, so a first component is
    $\actcount(a_{k-1 \bmod n})\land \somepath\xnext\
    (\Lor_{\scriptsize b\in \init(t_{k-1 \bmod n})\setminus\{\theta\}}
    \fttoctl(b)) \land\fttoctl(t_{k-1 \bmod n})$.  Note that any
    $\theta$ top-level branch in $t_{k-1 \bmod n}$ is invalidated
    (since some action $b\in \init(t_{k-1 \bmod n})$ is available).
    The top-level $\theta$ branch of $t_{k-1 \bmod n}$ can be taken
    only if no action from $\init(t_{k-1 \bmod n})\cup\{a_{k}\}$ is
    available, so the second variant is
    $\fttoctl(a_{k-1 \bmod n})\land\somepath\xnext\
    \lnot\fttoctl(a_{(k)})\land\lnot(\Lor_{\scriptsize b\in
      \init(t_{k-1 \bmod n})\setminus\{\theta\}}
    \fttoctl(b))\land\fttoctl(t_{k-1 \bmod n}({\theta}))$, where
    $t_{k-1 \bmod n}=t'\ \choice\ \theta; t_{k-1 \bmod n}(\theta)$ for
    some test $t'$ (recall that we can assume without loss of
    generality that there exists a single top-level $\theta$ branch in
    $t_{k-1 \bmod n}$).

    We take the disjunction of the above variants and so
    $E_{k-1 \bmod n}=\actcount(a_{k-1 \bmod n})\land\somepath\xnext
    (\Lor_{\scriptsize b\in \init(t_{k-1 \bmod n})\setminus\{\theta\}}
    \fttoctl(b)) \land\fttoctl(t_{k-1 \bmod n}) \lor
    \lnot\fttoctl(a_{(k)})\land\lnot(\Lor_{\scriptsize b\in
      \init(t_{k-1 \bmod n})\setminus\{\theta\}}
    \fttoctl(b))\land\fttoctl(t_{k-1 \bmod n}({\theta}))$.

  \item If $a_{k}=\theta$ and $\theta\in \init(t_{k-1 \bmod n})$, then
    we must modify the cycle as well as the exit test.  Let
    $B=\init(t_{k-1 \bmod n})\setminus\{\theta\}$.

    If an action from $B$ is available the loop cannot continue, so we
    replace in $C$ all occurrences of $a_{k}$ with
    $\lnot(\Lor_{b\in B} \fttoctl(b)$ so that
    $C_{i} =
    \somepath\globally(\fttoctl(a_{i})\land\somepath\xnext(\cdots
    \land\somepath\xnext(\fttoctl(\lnot(\Lor_{\scriptsize b\in
      \init(t_{k-1 \bmod n})\setminus\{\theta\}} \fttoctl(b)))
    \land\somepath\xnext\cdots\land\somepath\xnext(\fttoctl(a_{(i+n-1)\bmod
      n}))\cdots)))$.

    Similarly, when actions from $B$ are available the non-$\theta$
    component of the exit test is applicable, while the $\theta$
    branch can only be taken when no action from $B$ is offered.
    Therefore we have
    $E_{k-1 \bmod n} = \actcount(a_{k-1 \bmod n})\land
    \somepath\xnext\ \Lor_{\scriptsize b\in \init(t_{k-1 \bmod
        n})\setminus\{\theta\}} \fttoctl(b) \land \fttoctl(t_{k-1
      \bmod n}) \lor \lnot (\Lor_{\scriptsize b\in \init(t_{k-1 \bmod
        n})\setminus\{\theta\}} \fttoctl(b)) \land \fttoctl(t_{k-1
      \bmod n}({\theta}))$.  As before, $t_{k-1 \bmod n}({\theta})$ is
    the $\theta$-branch of $t_{k-1 \bmod n}$ that is,
    $t_{k-1 \bmod n}=t'\ \choice\ \theta; t_{k-1 \bmod n}(\theta)$ for
    some test $t'$.

    Finally, recall that originally
    $E_{k}=\actcount(a_{k})\land\somepath\xnext\ \fttoctl(t_{k})$.
    Now however $a_{k}=\theta$ and so we must apply to $E_{k}$ the
    same process that we repeatedly performed earlier namely, adding
    $\lnot (\Lor_{\scriptsize b\in \init(t_{k-1 \bmod
        n})\setminus\{\theta\}} \fttoctl(b))$.  In addition, $\theta$
    does not consume any input by definition, so the $\somepath\xnext$
    construction disappears.  In all we have
    $E_{k}=\actcount(a_{k-1})\land\lnot (\Lor_{\scriptsize b\in
      \init(t_{k-1 \bmod n})\setminus\{\theta\}} \fttoctl(b))\land
    \fttoctl(t_{k})$ (when $a_{k}$ = $\theta$, we do not need to
    count $a_{k}$).
  \end{enumerate}

  We now prove that the construction described above is correct.  We
  focus first on the initial, $\theta$-less formula.

  If the common actions are available for both $p$ and $t$ then
  $p\Suparrow {a_{i}}p_{1}\land
  p_{1}\Suparrow{a_{i+1}}\land\cdots\land
  p_{n-1}\Suparrow{a_{i+n-1}}p$, which shows that process $p$ performs
  some actions in the cycle. We further notice that these are
  equivalent to $\ltstokripke(p)\models a_{i}$ and
  $\ltstokripke(p_{1})\models
  a_{i+1}\land\cdots\land\ltstokripke(p_{n-1})\models a_{i+n-1}$,
  respectively.  Therefore
  $p\Suparrow{a_{i}}p_{1}\land p_{1}\Suparrow{a_{i+1}}\land\cdots\land
  p_{n-1}\Suparrow{a_{i+n-1}}p$ iff
  $\ltstokripke(p)\models\somepath\globally(\fttoctl(a_{i})\land\somepath\xnext(\fttoctl(a_{(i+1)\bmod
    n})\land \somepath\xnext\cdots\land
  \somepath\xnext(\fttoctl(a_{(i+n-1)\bmod n}))\cdots))$.  That is,
  \begin{equation}
    \label{eq 3}
      p\Suparrow{a_{i}}p_{1}\land
      p_{1}\Suparrow{a_{i+1}}\land\cdots\land
      p_{n-1}\Suparrow{a_{i+n-1}}p \mbox{ iff } \ltstokripke(p)\models C_{i}
  \end{equation}

  We exit from the cycle as follows: When
  $p\Suparrow {a_{i}}p'\not \Suparrow {a_{i+1}}$ then the process $p$
  must take the test $t_{i}$ after performing $a_{i}$ and pass it. If
  however the action $a_{i+1}$ is available in the cycle as well as in
  the test $t_{i}$, then it depends on the process $p$ whether it will
  continue in the cycle or will take the test $t_{i}$. That means $p$
  may take $t_{i}$ and pass the test or it may decide to continue in
  the cycle. Eventually however the process must take one of the exit
  tests.  Given the nature of may-testing one successful path is
  enough for $p$ to pass $t$.

  Formally, we note that $p\Suparrow {a_{i}}p' \land p' \may\ t_{i}$
  is equivalent to
  $\ltstokripke(p)\models
  \actcount(a_{i})\land\somepath\xnext(\fttoctl(t_{i}))$.  Indeed,
  when $p$ and $t$ perform $a_{i}$ the test $t$ is $i$ actions away
  from $\start(a_{0})$ and so $\actcount(a_{i})$ is true.  Therefore
  given the inductive hypothesis (that $p'\ \may\ t_{i}$ iff
  $\ltstokripke(p')\models\fttoctl(t_{i})$ for any process $p'$) we
  concluded that:
  \begin{equation}
    \label{eq 4}
    (p\Suparrow{a_{i}}p' \land p'\ \may\ t_{i}) \mbox{ iff } 
    \ltstokripke(p)\models(\actcount(a_{i}) \land\somepath\xnext(\fttoctl(t_{i}))) = E_{i}
  \end{equation}
  Taking the disjunction of Relations~(\ref{eq 4}) over all
  $0\leq i< n$ we have
  \begin{equation}
    \label{eq 5}
    (p\Suparrow{a_{i}}p' \land p'\ \may\ \mchoice_{i=0}^{n-1} t_{i}) \mbox{ iff }  \ltstokripke(p)\models \Lor_{i=0}^{n-1} E_{i}
  \end{equation} 
  The correctness of $\fttoctl(t)$ is then the direct consequences of
  Relations~(\ref{eq 3}) and~(\ref{eq 5}): We can stay in the cycle as
  long as one $C_{i}$ remains true (Relation~(\ref{eq 3})), and we can
  exit at any time using the appropriate exit test (Relation~(\ref{eq
    5})).

  Next we consider the possible deadlock detection action as
  introduced in the three cases above.  We have:

  \begin{enumerate}
  \item Let $a_{k}=\theta$, $\theta\not\in\init(t_{k-1 \bmod n})$ and
    suppose that $p\parallelth t$ runs along such that they reach the
    point
    $p'\parallelth t'\suparrow{a_{k-1 \bmod n}}p''\parallelth t''$.
    Let $b\in \init(t_{k-1 \bmod n})$.  If
    $p''\parallelth t''\suparrow{b}$ then the run must exit the cycle
    according to the definition of $\parallelth$.  At the same time
    $C_{k}$ is false because the disjunction
    $\Lor_{\scriptsize b\in \init(t_{k-1 \bmod n})} \fttoctl(b)$ is true
    and no other $C_{i}$ is true, so $C$ is false and therefore the
    only way for $\fttoctl(t)$ to be true is for $E_{k}$ to be true.
    The two, testing and logic scenarios are clearly equivalent.  On
    the other hand, if $p''\parallelth t''\not\suparrow{b}$ for any
    $b\in \init(t_{k-1 \bmod n})$, then the test must take the
    $\theta$ branch.  At the same time $C_{k}$ is true and so is $C$,
    whereas $E_{k}$ is false (so the formula must ``stay in the
    cycle''), again equivalent to the test scenario.
  
  \item Now $\theta\in\init(t_{k-1 \bmod n})$ and $a_{k}\neq \theta$.
    The way the process and the test perform $a_{k}$ and remain in the
    cycle is handled by the general case so we are only considering
    the exit test $t_{k-1 \bmod n}$.  The only supplementary
    consequence of $a_{k}$ being available is that any $\theta$ branch
    in $t_{k-1 \bmod n}$ is disallowed, which is still about the exit
    test rather than the cycle.

    There are two possible successful runs that involve the exit test
    $t_{k-1 \bmod n}$. First,
    $p\Suparrow {a_{k-1 \bmod n}}p'\land \exists b\in \init(t_{k-1
      \bmod n})\setminus\{\theta\}: p'\Suparrow{b} \land\ p'\ \may\
    t_{k-1 \bmod n}$. Second,
    $p\Suparrow {a_{k-1 \bmod n}}p'\land \lnot(\exists b\in
    \init(t_{k-1 \bmod n})\setminus\{\theta\}: p'\Suparrow{b}) \land
    p'\not\Suparrow{a_{k}} \land\ p'\ \may\ t_{k-1 \bmod n}(\theta)$.
    The first case corresponds to a common action $b$ being available
    to both the process and the test (case in which the $\theta$
    branch of $t_{k-1 \bmod n}$ is forbidden by the semantics of
    $p'\ \may\ t_{k-1 \bmod n}$). The second case requires that
    the $\theta$ branch of the test is taken whenever no other action
    is available.

    Given the inductive hypothesis (that $p'\ \may\ t_{i}$ iff
    $\ltstokripke(p')\models\fttoctl(t_{i})$ for any process $p'$) we
    have
    \begin{equation}
      \label{eq 6}
      \begin{split}
        p\Suparrow {a_{k-1 \bmod n}}p'\land \exists b\in
        \init(t_{k-1 \bmod n})\setminus\{\theta\}: p'\Suparrow{b} \land p'\
        \may\
        t_{k-1 \bmod n} \mbox{ iff }\\
        \ltstokripke(p)\models \actcount(a_{k-1 \bmod n})\land\somepath\xnext\
        \mbox{$\Lor_{\scriptsize b\in
            \init(t_{k-1 \bmod n})\setminus\{\theta\}}$} \fttoctl(b)
        \land\fttoctl(t_{k-1 \bmod n})
      \end{split}
    \end{equation}
    \begin{equation}
      \label{eq 6 1}
      \begin{split}
        p\Suparrow {a_{k-1 \bmod n}}p'\land \lnot(\exists b\in
        \init(t_{k-1 \bmod n})\setminus\{\theta\}: p'\Suparrow{b}) \land
        p'\not\Suparrow{a_{k}} \land \\ p'\ \may\ t_{k-1 \bmod n}(\theta) \mbox{ iff }
        \ltstokripke(p)\models \actcount(a_{k-1 \bmod n})\land\somepath\xnext\
        \lnot\fttoctl(a_{(k)})\land\\\lnot(\mbox{$\Lor_{\scriptsize b\in
            \init(t_{k-1 \bmod n})\setminus\{\theta\}}
          \fttoctl(b)$})\land\fttoctl(t_{k-1 \bmod n}({\theta}))
      \end{split}
    \end{equation}
    The conjunction of Relations~(\ref{eq 6}) and~(\ref{eq 6 1})
    establish this case.  Indeed, the left hand sides of the two
    relations are the only two ways to have a successful run involving
    $t_{k-1 \bmod n}$ (as argued above).

  \item Let now $a_{k}=\theta$ and $\theta\in \init(t_{k-1 \bmod n})$.
    Suppose that the process under test is inside the cycle and has
    reached a state $p$ such that $p\Suparrow{a_{k-1 \bmod n}} p'$,
    meaning that $p'$ is ready to either continue within the cycle or
    pass $t_{k-1 \bmod n}$.

    Suppose first that $p'\Suparrow{b}$ for some
    $b\in \init(t_{k-1 \bmod n})\setminus\{\theta\}$.  Then $(a)$ $p'$
    cannot continue in the cycle, which is equivalent to $C_{k}$ being
    false (since no $C_{i}$, $i\neq k$ can be true), and so $(b)$ $p'$
    must pass $t_{k-1 \bmod n}$, which is equivalent to
    $\ltstokripke(p')\models \Lor_{\scriptsize b\in \init(t_{k-1 \bmod
        n})\setminus\{\theta\}} \fttoctl(b) \land \fttoctl(t_{k-1
      \bmod n})$.  That $C_{k}$ is false happens because
    $\lnot(\Lor_{\scriptsize b\in \init(t_{k-1 \bmod
        n})\setminus\{\theta\}} \fttoctl(b))$ is false.  Note
    incidentally that the $\theta$ branch of $t_{k-1 \bmod n}$ is
    forbidden, but this is guaranteed by the semantics of $p'$ passing
    $t_{k-1 \bmod n}$ (and therefore by the semantics of
    $\ltstokripke(p')\models \fttoctl(t_{k-1 \bmod n})$ by inductive
    hypothesis).

    Suppose now that $p'\not\Suparrow{b}$ for any
    $b\in \init(t_{k-1 \bmod n})\setminus\{\theta\}$.  Then the only
    possible continuations are $(a)$ $p'$ remaining in the cycle which
    is equivalent to $C_{k}$ being true (ensured by
    $\Lor_{\scriptsize b\in \init(t_{k-1 \bmod n})\setminus\{\theta\}}
    \fttoctl(b)$ being false), or $(b)$ $p'$ taking the $\theta$
    branch of $t_{k-1 \bmod n}$, which is equivalent to
    $\lnot (\Lor_{\scriptsize b\in \init(t_{k-1 \bmod
        n})\setminus\{\theta\}} \fttoctl(b)) \land
    \fttoctl(t_{k}({\theta}))$ by the fact that
    $\Lor_{\scriptsize b\in \init(t_{k-1 \bmod n})\setminus\{\theta\}}
    \fttoctl(b)$ is false and the inductive hypothesis, or $(c)$
    $p'$ taking the test $t_{k}$ (which falls just after
    $a_{k}=\theta$ and so it is an alternative in the deadlock
    detection branch), which is equivalent to $E_{k}$ being true,
    ensured by
    $\Lor_{\scriptsize b\in \init(t_{k-1 \bmod n})\setminus\{\theta\}}
    \fttoctl(b)$ being false and $\fttoctl(t_{k})$ being true iff
    $p'$ passes $t_{k}$ by inductive hypothesis.

    Once more taking the disjunction of the two alternatives above
    establishes this case.
\end{enumerate}
\end{proof}

\newpage
\begin{qexample}{Compact CTL formula generation}

  In this section we illustrate the conversion of failure trace test
  into CTL formulae.  For this purpose consider the following simple
  vending machines:
  \begin{eqnarray*}
    P_{1} &=& \coin; (\coffee\ \choice\ \coin;\ (\tea\ \choice\ \bang;(\tea\ \choice\ \ P_{1})))\\
    P_{2} &=& \coin; (\coffee\ \choice\ \coin;\ (\coffee\ \choice\ \bang;(\tea\ \choice\ \ P_{2})))
  \end{eqnarray*}
  These two machines dispense a coffee after accepting a coin; the
  first machine dispenses a tea after second coin, whereas the second
  still dispenses coffee. After two coins and a hit by customer, the
  two machines will dispense tea.  The vending machines as well as the
  equivalent Kripke structures are shown in
  Figure~\ref{fig-vmachines}.

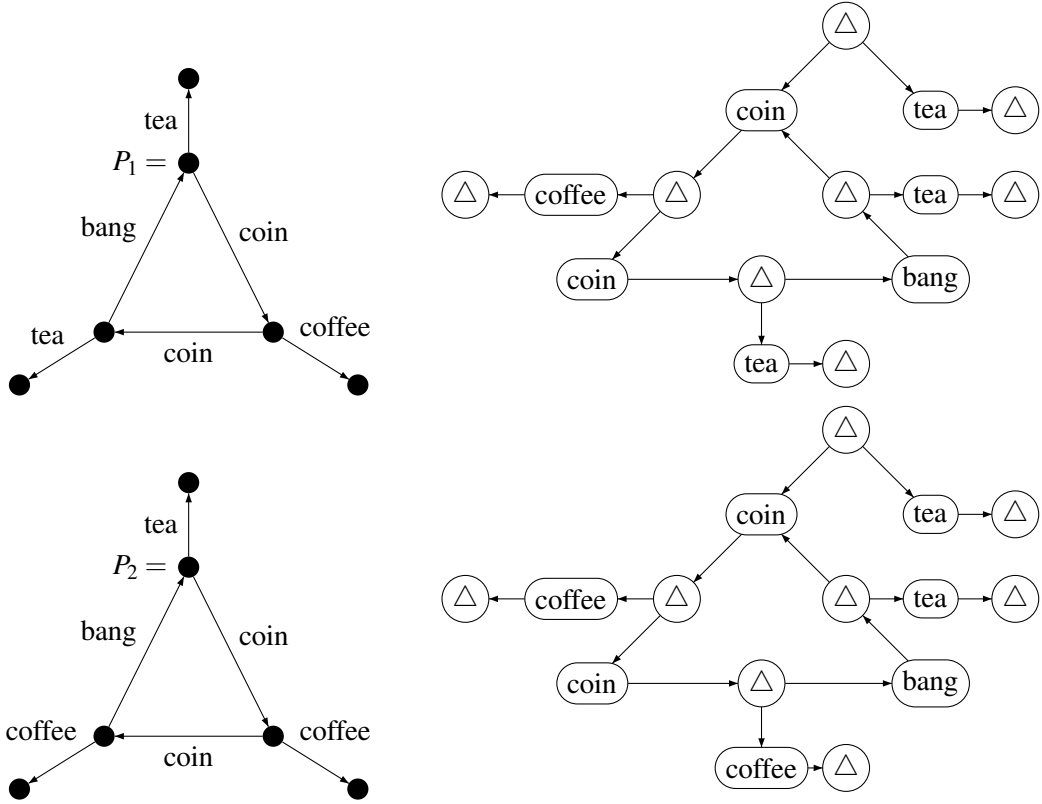
\begin{figure}[tb]
  \centering
  \begin{picture}(34,31)(-1,-6)
    \gasset{Nframe=n,ExtNL=y,NLdist=1,NLangle=180,Nw=2,Nh=2,Nframe=n,Nfill=y}
    \node[NLangle=180](p1)(16,16){$P_{1}=$}
    \node(p2)(24,0){}
    \node(p3)(8,0){}
    \drawtrans(p1,p2){coin}
    \drawtrans(p2,p3){coin}
    \drawtrans(p3,p1){bang}
    \node(p11)(16,24){}
    \node(p21)(32,-5){}
    \node(p31)(0,-5){}
    \drawtrans(p1,p11){tea}
    \drawtrans(p2,p21){coffee}
    \drawtrans[r](p3,p31){tea}  
  \end{picture}
  \qquad
  \begin{picture}(58,38)(-15,-11)
    \gasset{Nframe=y,ExtNL=n,NLdist=0,NLangle=180,Nw=1,Nh=1,Nframe=y,Nadjust=wh,Nadjustdist=1,Nfill=n}

    \node(m)(24,24){$\triangle$}

    \node(c1)(16,16){coin}
    \node(c1t)(8,8){$\triangle$}
    \node(c2)(0,0){coin}
    \node(c2t)(16,0){$\triangle$}
    \node(c3)(32,0){bang}
    \node(c3t)(24,8){$\triangle$}

    \node(xm)(32,16){tea}
    \node(xmt)(40,16){$\triangle$}

    \node(xc1)(-2,8){coffee}
    \node(xct1)(-12,8){$\triangle$}

    \node(xc2)(16,-8){tea}
    \node(xc2t)(24,-8){$\triangle$}

    \node(xc3)(32,8){tea}
    \node(xc3t)(40,8){$\triangle$}

    \drawtrans(m,c1){}

    \drawtrans(c1,c1t){}
    \drawtrans(c1t,c2){}
    \drawtrans(c2,c2t){}
    \drawtrans(c2t,c3){}
    \drawtrans(c3,c3t){}
    \drawtrans(c3t,c1){}

    \drawtrans(m,xm){}
    \drawtrans(xm,xmt){}

    \drawtrans(c1t,xc1){}
    \drawtrans(xc1,xct1){}

    \drawtrans(c2t,xc2){}
    \drawtrans(xc2,xc2t){}

    \drawtrans(c3t,xc3){}
    \drawtrans(xc3,xc3t){}
  \end{picture}

  \begin{picture}(34,31)(-1,-6)
    \gasset{Nframe=n,ExtNL=y,NLdist=1,NLangle=180,Nw=2,Nh=2,Nframe=n,Nfill=y}
    \node[NLangle=180](p1)(16,16){$P_{2}=$}
    \node(p2)(24,0){}
    \node(p3)(8,0){}
    \drawtrans(p1,p2){coin}
    \drawtrans(p2,p3){coin}
    \drawtrans(p3,p1){bang}
    \node(p11)(16,24){}
    \node(p21)(32,-5){}
    \node(p31)(0,-5){}
    \drawtrans(p1,p11){tea}
    \drawtrans(p2,p21){coffee}
    \drawtrans[r](p3,p31){coffee}  
  \end{picture}
  \qquad
  \begin{picture}(58,38)(-15,-11)
    \gasset{Nframe=y,ExtNL=n,NLdist=0,NLangle=180,Nw=1,Nh=1,Nframe=y,Nadjust=wh,Nadjustdist=1,Nfill=n}

    \node(m)(24,24){$\triangle$}

    \node(c1)(16,16){coin}
    \node(c1t)(8,8){$\triangle$}
    \node(c2)(0,0){coin}
    \node(c2t)(16,0){$\triangle$}
    \node(c3)(32,0){bang}
    \node(c3t)(24,8){$\triangle$}

    \node(xm)(32,16){tea}
    \node(xmt)(40,16){$\triangle$}

    \node(xc1)(-2,8){coffee}
    \node(xct1)(-12,8){$\triangle$}

    \node(xc2)(16,-8){coffee}
    \node(xc2t)(24,-8){$\triangle$}

    \node(xc3)(32,8){tea}
    \node(xc3t)(40,8){$\triangle$}

    \drawtrans(m,c1){}

    \drawtrans(c1,c1t){}
    \drawtrans(c1t,c2){}
    \drawtrans(c2,c2t){}
    \drawtrans(c2t,c3){}
    \drawtrans(c3,c3t){}
    \drawtrans(c3t,c1){}

    \drawtrans(m,xm){}
    \drawtrans(xm,xmt){}

    \drawtrans(c1t,xc1){}
    \drawtrans(xc1,xct1){}

    \drawtrans(c2t,xc2){}
    \drawtrans(xc2,xc2t){}

    \drawtrans(c3t,xc3){}
    \drawtrans(xc3,xc3t){}
  \end{picture}

  \caption{Yet another couple of vending machines $P_{1}$ and $P_{2}$
    (left) and the equivalent Kripke structures $\ltstokripkex(P_{1})$
    and $\ltstokripkex(P_{2})$ (right).}
  \label{fig-vmachines}
\end{figure}

Consider now the following test:
\[  
  t=\coin;(\coffee;\pass\ \choice\ \coin;(\tea;\pass\ \choice\ \bang;(\tea;\pass\ \choice\ t))) 
\]
Using Theorem~\ref{th-ctl-ftr-compt} (and thus implicitly
Lemma~\ref{th-ft-ctl-x}) we can convert this test into the following
CTL formula (after eliminating all the obviously true sub-formulae):
\[
  \begin{array}{rcl}
    \fttoctl_{\ltstokripkex}(t) &=& (\somepath\globally\ (\coin)\land\somepath\xnext(\coin\land\somepath\xnext((\bang)\land\somepath\xnext(\coin)))\lor\\ 
&&\somepath\globally\ (\coin)\land\somepath\xnext(\bang\land\somepath\xnext((\coin)\land\somepath\xnext(\coin)))\lor\\
&&\somepath\globally\ (\bang)\land\somepath\xnext(\coin\land\somepath\xnext((\coin)\land\somepath\xnext(\bang)))\\
                     && \until\ \somepath\globally\ \start(\coin) \land\somepath\xnext\ (\coffee)\lor\somepath
\globally\ \start(\coin) \\
&&\land\ \somepath\xnext(\coin\land\somepath\xnext(\tea))\lor\somepath\globally\ \start(\coin) \land\somepath\xnext(\coin\land\somepath\xnext(\bang\land\somepath\xnext(\tea)))) \\
&& \lor\ \somepath\globally(\tea)
  \end{array}
\]

It is not difficult to see that the meaning of this formula is
equivalent to the meaning of $t$. Indeed, the following is true for
both the test $t$ as well as the formula
$\fttoctl_{\ltstokripkex}(t)$: tea is offered in the first step without
any coin, or after a coin coffee is offered, or after two coins tea is
offered, or after two coins and customer hits tea is offered; On the
other hand after two coins and a bang, and if no coin is available
next, then tea will be offered. Finally, after two coins and a bang if
both coin and tea are available, then two options comes out: we can
either keep in the cycle, or offer a tea. In all the process can
remain in the cycle indefinitely or can exit from the cycle and pass
the test. We can get coffee or tea at the first cycle or after few
repetitions of the cycle.

The formula $\fttoctl_{\ltstokripkex}(t)$ holds for
$\ltstokripkex(P_{1})$ but it does not hold for
$\ltstokripkex(P_{2})$.
\end{qexample}

\subsection{From CTL Formulae to Failure Trace Tests}
\label{sect-ctl-to-ft}

We find it convenient to show first how to consider logical (but not
temporal) combinations of tests.

\begin{lemma}
  \label{th-ft-neg}
  For any test $t\in \ftset$ there exists a test $\nottest{t}\in
  \ftset$ such that $p\ \may\ t$ if and only if $\neg(p\ \may\ \nottest{t})$ for
  any $p\in\procset$.
\end{lemma}

\begin{proof}
  We modify $t$ to produce \nottest{t} as follows: We first force all
  the success states to become deadlock states by eliminating the
  outgoing action $\gamma$ from $t$ completely.  If the action that
  leads to a state thus converted is $\theta$ then this action is
  removed as well (indeed, the ``fail if nothing else works''
  phenomenon thus eliminated is implicit in testing).  Finally, we add
  to all the states having their outgoing transitions $\theta$ or
  $\gamma$ eliminated in the previous steps one outgoing transition
  labeled $\theta$ followed by one outgoing transition labeled
  $\gamma$.

  The test $\nottest{t}$ must fail every time the original test $t$
  succeeds.  The first step (eliminating $\gamma$ transitions) has
  exactly this effect.

  In addition, we must ensure that $\nottest{t}$ succeeds in all the
  circumstances in which $t$ fails.  The extra $\theta$ outgoing
  actions (followed by $\gamma$) ensure that whevever the end of the
  run reaches a state that was not successful in $t$ (meaning that it
  had no outgoing $\gamma$ transition) then this run is extended in
  $\nottest{t}$ via the $\theta$ branch to a success state, as
  desired.
\end{proof}

\begin{lemma}
  \label{th-ft-or}
  For any two tests $t_{1}, t_{2}\in \ftset$ there exists a test
  $t_{1}\lor t_{2}\in \ftset$ such that $p\ \may\ (t_{1}\lor t_{2})$
  if and only if $(p\ \may\ t_{1})\lor (p\ \may\ t_{2})$ for any $p\in \procset$.
\end{lemma}

\begin{proof}
  We construct such a disjunction on tests by induction over the
  structure of tests.

  For the base case it is immediate that $\pass \lor t = t\lor \pass =
  \pass$ and $\pstop \lor t = t \lor \pstop = t$ for any structure of
  the test $t$.

  For the induction step we consider without loss of generality that
  $t_{1}$ and $t_{2}$ have the following structure:
  \begin{eqnarray*}
    t_{1} &=& \mchoice \{b_1; t_1(b_1):b_1\in B_1 \}\ \choice\ \theta; t_{N1}\\
    t_{2} &=& \mchoice \{b_2; t_2(b_2):b_2\in B_2 \}\ \choice\ \theta; t_{N2}
  \end{eqnarray*}
  Indeed, all the other possible structures of $t_{1}$ (and $t_{2}$)
  are covered by such a form since $\theta$ not appearing on the top
  level of $t_{1}$ (or $t_{2}$) is equivalent to $t_{N1} = \pstop$ (or
  $t_{N2} = \pstop$), while not having a choice on the top level of
  the test is equivalent to $B_{1}$ (or $B_{2}$) being an appropriate
  singleton.

  We then construct $t_{1}\lor t_{2}$ for the form of $t_{1}$ and
  $t_{2}$ mentioned above under the inductive assumption that the
  disjunction between any of the ``inner'' tests $t_1(b_1)$, $t_{N1}$
  $t_2(b_2)$, and $t_{N2}$ is known.  We have:
  \begin{equation}
    \begin{array}{rclc}
      \multicolumn{4}{l}{
        \mchoice \{b_1; t_1(b_1):b_1\in B_1 \}\ \choice\ \theta; t_{N1}\quad  \lor\quad
        \mchoice \{b_2; t_2(b_2):b_2\in B_2 \}\ \choice\ \theta; t_{N2}}\\
      \qquad\qquad &=& \mchoice \{b; (t_1 (b) \vee t_2 (b)): b \in B_1 \cap B_2\} & \choice \\
      & & \mchoice \{b; (t_1 (b) \lor \restrict{t_{N2}}{b}) : b \in B_1 \setminus B_{2} \} & \choice \\
      & & \mchoice \{b; (t_2 (b) \lor \restrict{t_{N1}}{b}): b \in B_2 \setminus B_{1} \} & \choice\\
      & & \theta; (t_{N1} \vee t_{N2})
    \end{array}
    \label{eq-ft-or}
  \end{equation}
  where $\restrict{t}{b}$ is the test $t$ restricted to performing $b$
  as its first action and so is inductively constructed as follows:
  \begin{enumerate}
  \item If $t = \pstop$ then $\restrict{t}{b} = \pstop$.
  \item If $t = \pass$ then $\restrict{t}{b} = \pass$.
  \item If $t = b; t'$ then $\restrict{t}{b} = t'$.
  \item If $t = a; t'$ with $a\neq b$ then $\restrict{t}{b} = \pstop$.
  \item If $t = \internal; t'$ then $\restrict{t}{b} =
    \restrict{t'}{b}$.
  \item If $t =t'\ \choice\ t''$ such that neither $t'$ nor $t''$ contain
    $\theta$ in their topmost choice then $\restrict{t}{b} =
    \restrict{t'}{b}\ \choice\ \restrict{t''}{b}$.
  \item If $t =b; t'\ \choice\ \theta; t''$ then $\restrict{t}{b} = t'$.
  \item If $t =a; t'\ \choice\ \theta; t''$ with $a\neq b$ then
    $\restrict{t}{b} = \restrict{t''}{b}$.
  \end{enumerate}

  If the test $t_{1}\lor t_{2}$ is offered an action $b$ that is
  common to the top choices of the two tests $t_{1}$ and $t_{2}$
  ($b\in B_{1}\cap B_{2}$) then the disjunction succeeds if and only if $b$ is
  performed and then at least one of the tests $t_{1}(b)$ and
  $t_{2}(b)$ succeeds (meaning that $t_{1}(b) \lor t_{2}(b)$ succeeds
  inductively) afterward.  The first term of the choice in
  Equation~(\ref{eq-ft-or}) represents this possibility.

  If the test is offered an action $b$ that appears in the top choice
  of $t_{1}$ but not in the top choice of $t_{2}$ ($b \in B_1
  \setminus B_{2}$) then the disjunction succeeds if and only if $t_{1}(b)$
  succeeds after $b$ is performed, or $t_{N2}$ performs $b$ and then
  succeeds; the same goes for $b \in B_2 \setminus B_{1}$ (only in
  reverse).  The second and the third terms of the choice in
  Equation~(\ref{eq-ft-or}) represent this possibility.

  Finally whenever the test $t_{1}\lor t_{2}$ is offered an action $b$
  that is in neither the top choices of the two component tests
  $t_{1}$ and $t_{2}$ (that is, $b\not\in B_{1}\cup B_{2}$), then the
  disjunction succeeds if and only if at least one of the tests $t_{N1}$ or
  $t_{N2}$ succeeds (or equivalently $t_{N1}\lor t_{N2}$ succeeds
  inductively).  This is captured by the last term of the choice in
  Equation~(\ref{eq-ft-or}).  Indeed, $b\not\in B_{1}\cup B_{2}$
  implies that $b\not\in (B_{1}\cap B_{2}) \cup (B_{1}\setminus B_{2})
  \cup (B_{2}\setminus B_{1})$ and so such an action will trigger the
  deadlock detection ($\theta$) choice.

  There is no other way for the test $t_{1}\lor t_{2}$ to succeed so
  the construction is complete. 
\end{proof}

We believe that an actual example of how disjunction is constructed is
instructive.  The following is therefore an example to better
illustrate disjunction over tests.  A further example (incorporating
temporal operators and also negation) will be provided later (see
Example~\ref{ex-coffee-1}).

\begin{qexample}{A disjunction of tests}
  Consider the construction $t_{1}\lor t_{2}$, where:
  \begin{eqnarray*}
    t_{1} &=& (\bang; \tea; \pass)\ \choice\\
          && (\coin;  (\coffee; \pstop\ \choice\ \theta; \pass))
  \end{eqnarray*}
  \begin{eqnarray*}
    t_{2} &=& (\coin; \pass)\ \choice\\
    &&(\nudge; \pstop)\ \choice \\
    && (\theta; (\bang; \water; \pass\ \choice\ \turn; \pstop))
  \end{eqnarray*}
  Using the notation from Equation~(\ref{eq-ft-or}) we have
  $B_{1}=\{\coin, \bang\}$, $B_{2}=\{\coin, \nudge\}$ and so
  $B_{1}\cap B_{2} = \{\coin\}$, $B_{1}\setminus B_{2} = \{\bang\}$,
  and $B_{2}\setminus B_{1} = \{\nudge\}$.  We further note that
  $t_{1}(\coin) = \coffee; \pstop\ \choice\ \theta; \pass$,
  $t_{1}(\bang) = \tea; \pass$, $t_{N1} = \pstop$, $t_{2}(\coin) =
  \pass$, $t_{2}(\nudge) = \pstop$, and $t_{N2} = \bang; \water;
  \pass\ \choice\ \turn; \pstop$.  Therefore we have:
  \begin{eqnarray*}
    t_{1}\lor t_{2} &=& (\coin; (t_{1}(\coin) \lor t_{2}(\coin)))\ \choice\\
                     && (\bang; (t_{1}(\bang) \lor \restrict{t_{N2}}{\bang})) \ \choice\ \\
                    & & (\nudge; (t_{2}(\nudge) \lor \restrict{t_{N1}}{\nudge})) \ \choice\\
                     && \theta; (t_{N1}\lor t_{N2})\\
    &=& (\coin; \pass) \ \choice \\
     && (\bang; (t_{1}(\bang) \lor \restrict{t_{N2}}{\bang}))\ \choice\ \\
    & & (\nudge; \pstop)\ \choice\\
     && (\theta; (\bang; \water; \pass\ \choice\ \turn; \pstop))
  \end{eqnarray*}
  Indeed, $t_{1}(\coin) = \pass$ so $t_{1}(\coin)\lor t_{2}(\coin) =
  \pass$; $\restrict{t_{N1}}{\nudge} = \pstop$; and $t_{N1} = \pstop$
  so $t_{N1}\lor t_{N2} = t_{N2}$.

  We further have $\restrict{t_{N2}}{\bang} = \water; \pass$, and
  therefore $t_{1}(\bang) \lor \restrict{t_{N2}}{\bang} = (\tea;
  \pass) \lor (\water; \pass)$.  We proceed inductively as above,
  except that in this degenerate case $t_{N1} = t_{N2} = \pstop$ and
  $B_{1}\cap B_{2} = \emptyset$, so the result is a simple choice
  between the components: $t_{1}(\bang) \lor \restrict{t_{N2}}{\bang}
  = (\tea; \pass)\ \choice\ (\water; \pass)$.  Overall we reach the
  following result:

  \begin{eqnarray}
    t_{1}\lor t_{2} &=& (\coin; \pass) \ \choice \nonumber \\
                     && (\bang; ((\tea; \pass)\ \choice\ (\water; \pass)))\ \choice\ \nonumber \\
                    & & (\nudge; \pstop)\ \choice \nonumber \\
                     && (\theta; (\bang; \water; \pass\ \choice\ \turn; \pstop))
                        \label{eq-ex2-1}
  \end{eqnarray}
  Intuitively, $t_{1}$ specifies that we can have tea if we hit the
  machine, and if we put a coin in we can get anything except coffee.
  Similarly $t_{2}$ specifies that we can put a coin in the machine,
  we cannot nudge it, and if none of the above happen then we can hit
  the machine (case in which we get water) but we cannot turn it
  upside down.  On the other hand, the disjunction of these tests as
  shown in Equation~(\ref{eq-ex2-1}) imposes the following
  specification:
  \begin{enumerate}
  \item If a coin is inserted then the test succeeds.  Indeed, this
    case from $t_{2}$ supersedes the corresponding special case from
    $t_{1}$: the success of the test implies that the machine is
    allowed to do anything afterward, including not dispensing coffee.
  \item We get either tea or water after hitting the machine.  Getting
    tea comes from $t_{1}$ and getting water from $t_{2}$ (where the
    \bang\ event comes from the $\theta$ branch).
  \item If we nudge the machine then the test fails immediately; this
    comes directly from $t_{2}$.
  \item In all the other cases the test behaves like the $\theta$
    branch of $t_{2}$.  This behaviour makes sense since there is no
    such a branch in $t_{1}$.
  \end{enumerate}
  We can thus see how the disjunction construction from
  Lemma~\ref{th-ft-or} makes intuitive sense.

  This all being said, note that we do not claim that either of the
  tests $t_{1}$ and $t_{2}$ are useful in any way, and so we should
  not be held responsible for the behaviour of any machine built
  according to the disjunctive specification shown above.
\end{qexample}

\begin{corollary}
  \label{th-ft-and}
  For any tests $t_{1}, t_{2}\in \ftset$ there exists a test
  $t_{1}\land t_{2}\in\ftset$ such that $p\ \may\ (t_{1}\land t_{2})$
  if and only if $(p\ \may\ t_{1})\land (p\ \may\ t_{2})$ for any $p\in\procset$.
\end{corollary}

\begin{proof}
  Immediate from Lemmata~\ref{th-ft-neg} and~\ref{th-ft-or} using the
  De~Morgan rule $a\land b = \neg(\neg a \lor \neg b)$. 
\end{proof}

\medskip
\noindent
We are now ready to show that any CTL formula can be converted into an
equivalent failure trace test.

\begin{lemma}
  \label{th-ctl-ft}
  There exists a function $\ctltoft_{\ltstokripke}:\ctlset\rightarrow \ftset$ such
  that $\ltstokripke(p) \sat f$ if and only if $p\ \may\ \ctltoft_{\ltstokripke}(f)$ for any
  $p\in\procset$.
\end{lemma}

\begin{proof}
  The proof is done by structural induction over CTL formulae.  As
  before, the function $\ctltoft_{\ltstokripke}$ will also be defined recursively in
  the process.

  We have naturally $\ctltoft_{\ltstokripke}(\top) = \pass$ and
  $\ctltoft_{\ltstokripke}(\bot)=\pstop$.  Clearly any Kripke structure satisfies
  $\top$ and any process passes $\pass$, so it is immediate that
  $\ltstokripke(p)\sat \top$ if and only if $p\ \may\ \pass = \ctltoft_{\ltstokripke}(\top)$.
  Similarly $\ltstokripke(p)\sat \bot$ if and only if $p\ \may\ \pstop$ is
  immediate (neither is ever true).

  To complete the basis we have $\ctltoft_{\ltstokripke}(a) = a; \pass$, which is an
  immediate consequence of the definition of $\ltstokripke$.  Indeed,
  the construction defined in Theorem~\ref{th-kripke-lts} ensures that
  for every outgoing action $a$ of an LTS process $p$ there will be an
  initial Kripke state in $\ltstokripke(p)$ where $a$ holds and so $p\
  \may\ a;\pass$ if and only if $\ltstokripke(p)\sat a$.

  The constructions for non-temporal logical operators have already
  been presented in Lemma~\ref{th-ft-neg}, Lemma~\ref{th-ft-or}, and
  Corollary~\ref{th-ft-and}.  We therefore have $\ctltoft_{\ltstokripke}(\lnot f) =
  \nottest{\ctltoft_{\ltstokripke}(f)}$ with $\nottest{\ctltoft_{\ltstokripke}(f)}$ as constructed
  in Lemma~\ref{th-ft-neg}, while $\ctltoft_{\ltstokripke}(f_{1}\lor f_{2}) =
  \ctltoft_{\ltstokripke}(f_{1}) \lor \ctltoft_{\ltstokripke}(f_{2})$ and $\ctltoft_{\ltstokripke}(f_{1}\land
  f_{2}) = \ctltoft_{\ltstokripke}(f_{1}) \land \ctltoft_{\ltstokripke}(f_{2})$ with test
  disjunction and conjunction as in Lemma~\ref{th-ft-or} and
  Corollary~\ref{th-ft-and}, respectively.  That these constructions
  are correct follow directly from the respective lemmata and
  corollary.

  We then move to the temporal operators.

  We have $\ctltoft_{\ltstokripke}(\somepath\xnext\ f) = \mchoice \{a; \ctltoft_{\ltstokripke}(f) :
  a \in A\}$.  When applied to some process $p$ the resulting test
  performs any action $a\in A$ and then (in the next state $p'$ such
  that $p\suparrow{a}p'$) gives control to $\ctltoft_{\ltstokripke}(f)$.
  $\ctltoft_{\ltstokripke}(f)$ will always test the next state (since there is no
  $\theta$ in the top choice), and there is no restriction as to what
  particular next state $p'$ is expected (since any action of $p$ is
  accepted by the test).

  Concretely, suppose that $p\ \may\ \mchoice \{a; \ctltoft_{\ltstokripke}(f) : a \in
  A\}$.  Then there exists some action $a\in A$ such that $p$ performs
  $a$, becomes $p'$, and $p'\ \may\ \ctltoft_{\ltstokripke}(f)$ (by definition of may
  testing).  It follows by inductive assumption that $\ltstokripke(p')
  \sat f$ and so $\ltstokripke(p) \sat \xnext\ f$ (since
  $\ltstokripke(p')$ is one successor of $\ltstokripke(p)$).
  Conversely, by the definition of $\ltstokripke$ all successors of
  $\ltstokripke(p)$ have the form $\ltstokripke(p')$ such that
  $p\suparrow{a}p'$ for some $a\in A$.  Suppose then that
  $\ltstokripke(p) \sat \xnext\ f$.  Then there exists a successor
  $\ltstokripke(p')$ of $\ltstokripke(p)$ such that $\ltstokripke(p')
  \sat f$, which is equivalent to $p'\ \may\ \ctltoft_{\ltstokripke}(f)$ (by
  inductive assumption), which in turn implies that $(p = a; p')\ \may\
  \mchoice \{a; \ctltoft_{\ltstokripke}(f) : a \in A\}$ (by the definition of may
  testing), as desired.

  \begin{figure}[tbp]
    \centering
    \begin{tabular}{ccc}
      \mbox{\begin{picture}(21,35)(-6,14)
          \gasset{Nframe=n,ExtNL=y,NLdist=1,NLangle=180,Nw=2,Nh=2,Nframe=n,Nfill=y}

          \node(t1)(0,46){}
          \rpnode[polyangle=90,iangle=-90,Nfill=n,Nframe=y,ExtNL=n,NLdist=0](rt1)(0,39)(3,6){\small\vspace{-3mm} $\ctltoft_{\ltstokripke}(f)$}
          \drawloop[r](t1){$a\in A$}
        \end{picture}}
      &\quad&
      \mbox{\begin{picture}(73,35)(-6,14)
          \gasset{Nframe=n,ExtNL=y,NLdist=1,NLangle=180,Nw=2,Nh=2,Nframe=n,Nfill=y}

          \node(t1)(0,46){}
          \rpnode[polyangle=90,iangle=-90,Nfill=n,Nframe=y,ExtNL=n,NLdist=0](rt1)(0,39)(3,6){\small\vspace{-3mm} $\ctltoft_{\ltstokripke}(f)$}

          \node[Nfill=n,Nframe=n,ExtNL=n,NLdist=0](t21)(15,22){}
          \rpnode[polyangle=90,iangle=-90,Nfill=n,Nframe=y,ExtNL=n,NLdist=0](rt2)(15,17)(3,6){\small\vspace{-3mm} $\ctltoft_{\ltstokripke}(f)$}

          \node(t31)(30,28){}
          \node[Nfill=n,Nframe=n,ExtNL=n,NLdist=0](t32)(30,22){}
          \rpnode[polyangle=90,iangle=-90,Nfill=n,Nframe=y,ExtNL=n,NLdist=0](rt3)(30,17)(3,6){\small\vspace{-3mm} $\ctltoft_{\ltstokripke}(f)$}

          \node(t4)(55,35){}
          \node[Nadjust=wh,Nadjustdist=0.5,Nfill=n,Nframe=n,ExtNL=n,NLdist=0](t41)(55,28){$\vdots$}
          \node[Nfill=n,Nframe=n,ExtNL=n,NLdist=0](t42)(55,22){}
          \rpnode[polyangle=90,iangle=-90,Nfill=n,Nframe=y,ExtNL=n,NLdist=0](rt4)(55,17)(3,6){\small\vspace{-3mm} $\ctltoft_{\ltstokripke}(f)$}

          \node[Nfill=n,Nframe=n,ExtNL=n,NLdist=0](dots)(42,17){\ldots}
          \node[Nfill=n,Nframe=n,ExtNL=n,NLdist=0](dots)(65,17){\ldots}

          \drawtrans[l](t1,t21){$a\in A$}
          \drawtrans[l](t1,t31){$a\in A$}
          \drawtrans[l](t31,t32){$a\in A$}
          \drawtrans[l](t1,t4){$a\in A$}
          \drawtrans[l](t4,t41){$a\in A$}
          \drawtrans[l](t41,t42){$a\in A$}
        \end{picture}}
      \\
      $(a)$ &\quad& $(b)$
    \end{tabular}
    \caption{Test equivalent to the CTL formula $\somepath\eventually\
      f$ $(a)$ and its unfolded version $(b)$.}
    \label{fig-omega-eff}
  \end{figure}
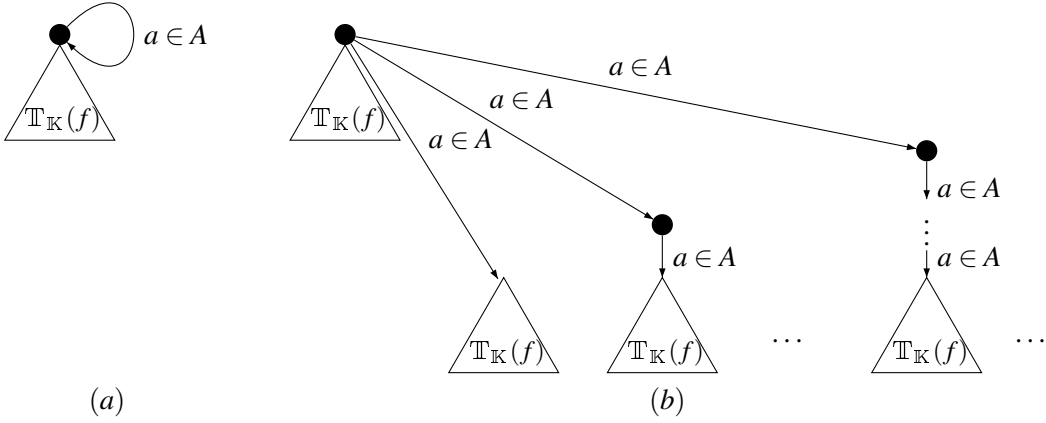

  We then have $\ctltoft_{\ltstokripke}(\somepath\eventually\ f) = t'$ such that $t'
  = \ctltoft_{\ltstokripke}(f)\ \choice\ (\mchoice (a;t' : a\in A))$.  The test $t'$
  is shown graphically in Figure~\ref{fig-omega-eff}$(a)$.  It
  specifies that at any given time the system under test has a choice
  to either pass $\ctltoft_{\ltstokripke}(f)$ or perform some (any) action and then
  pass $t'$ anew.  Repeating this description recursively we conclude
  that to be successful the system under test can pass $\ctltoft_{\ltstokripke}(f)$,
  or perform an action and then pass $\ctltoft_{\ltstokripke}(f)$, or perform two
  actions and then pass $\ctltoft_{\ltstokripke}(f)$, and so on.  The overall effect
  (which is shown in Figure~\ref{fig-omega-eff}$(b)$) is that exactly
  all the processes $p$ that perform an arbitrary sequence of actions
  and then pass $\ctltoft_{\ltstokripke}(f)$ at the end of this sequence will pass
  $t' = \ctltoft_{\ltstokripke}(\somepath\eventually\ f)$.  Given the inductive
  assumption that $\ctltoft_{\ltstokripke}(f)$ is equivalent to $f$ this is
  equivalent to $\ltstokripke(p)$ being the start of an arbitrary path
  to some state that satisfies $f$, which is precisely the definition
  of $\ltstokripke(p) \sat \somepath\eventually\ f$, as desired.

  Following a similar line of thought we have
  $\ctltoft_{\ltstokripke}(\somepath\globally\ f) = \ctltoft_{\ltstokripke}(f)\land
  (\ctltoft_{\ltstokripke}(\somepath\xnext\ f')\ \choice\ \theta; \pass)$, with $f' =
  f \land \somepath\xnext\ f'$.  Suppose that $p\ \may\
  \ctltoft_{\ltstokripke}(\somepath\globally\ f)$.  This implies that $p\ \may\
  \ctltoft_{\ltstokripke}(f)$.  This also implies that $p\ \may\
  \ctltoft_{\ltstokripke}(\somepath\xnext\ f')$ but only unless $p=\pstop$; indeed,
  the $\theta$ choice appears in conjunction with a multiple choice
  that offers all the possible alternatives (see the conversion for
  $\somepath\xnext$ above) and so can only be taken if no other action
  is available.

  By inductive assumption $p\ \may\ \ctltoft_{\ltstokripke}(f)$ if and only if
  $\ltstokripke(p)\sat f$.  By the conversion of $\somepath\xnext$
  (see above) $p\ \may\ \ctltoft_{\ltstokripke}(\somepath\xnext\ f')$ if and only if
  $\ltstokripke(p')\sat f'$ where $p\suparrow{a}p'$ for some $a\in A$
  and so $\ltstokripke(p')$ is the successor of $\ltstokripke(p)$ on
  some path.  We thus have $p\ \may\ \ctltoft_{\ltstokripke}(\somepath\globally\ f)$
  if and only if $\ltstokripke(p)\sat f$ and $\ltstokripke(p')\sat f'$ for some
  successor $p'$ of $p$.  Repeating the reasoning above recursively
  (starting from $p'$, etc.) we conclude that $p\ \may\
  \ctltoft_{\ltstokripke}(\somepath\globally\ f)$ if and only if $\ltstokripke(p_{1})\sat f$ for
  all the states $\ltstokripke(p_{1})$ on some path that starts from
  $\ltstokripke(p)$, which is clearly equivalent to
  $\ltstokripke(p)\sat \somepath\globally\ f'$.  The recursive
  reasoning terminates at the end of the path, when the LTS state $p$
  becomes $\pstop$, and the process is therefore released from its
  obligation to have states in which $f$ holds (since no states exist
  any longer).  In this case $\ltstokripke(p)$ is a ``sink'' state
  with no successor (according to Item~\ref{th-kripke-lts-3} in
  Theorem~\ref{th-kripke-lts}) and so the corresponding Kripke
  path is also at an end (and so there are no more states for $f$ to
  hold in).

  Finally we have $\ctltoft_{\ltstokripke}(\somepath\ f_{1}\ \until\ f_{2}) =
  (\ctltoft_{\ltstokripke}(f_{1}) \land (\ctltoft_{\ltstokripke}(\somepath\xnext\ f')\ \choice\
  \theta; \pass))\ \choice\ \internal; (\ctltoft_{\ltstokripke}(f_{2})\land
  (\ctltoft_{\ltstokripke}(\somepath\xnext\ f'')\ \choice\ \theta; \pass))$, with
  $f'=f_{1}\land\somepath\xnext\ f'$ and $f'' = f_{2}\land
  \somepath\xnext\ f''$.  Following the same reasoning as above (in
  the $\somepath\globally$ case) the fact that a process $p$ follows
  the test $\ctltoft_{\ltstokripke}(f_{1}) \land \ctltoft_{\ltstokripke}(\somepath\xnext\ f')$
  without deadlocking until some arbitrary state $p'$ is reached is
  equivalent to $\ltstokripke(p)$ featuring a path of arbitrary length
  ending in state $\ltstokripke(p')$ whose states
  $\ltstokripke(p_{1})$ will all satisfy $f_{1}$.  At the arbitrary
  (and nondeterministically chosen) point $p'$ the test
  $\ctltoft_{\ltstokripke}(\somepath\ f_{1}\ \until\ f_{2})$ will exercise its choice
  and so $p'$ has to pass $\ctltoft_{\ltstokripke}(f_{2})\land
  \ctltoft_{\ltstokripke}(\somepath\xnext\ f'')$ in order for $p$ to pass
  $\ctltoft_{\ltstokripke}(\somepath\ f_{1}\ \until\ f_{2})$.  We follow once more
  the same reasoning as in the $\somepath\globally$ case and we thus
  conclude that this is equivalent to $\ltstokripke(p_{1})\sat f$ for
  all the states $\ltstokripke(p_{1})$ on some path that starts from
  $\ltstokripke(p')$.  Putting the two phases together we have that
  $p\ \may\ \ctltoft_{\ltstokripke}(\somepath\ f_{1}\ \until\ f_{2})$ if and only if
  $\ltstokripke(p)$ features a path along which $f_{1}$ holds up to
  some state and $f_{2}$ holds from that state on, which is equivalent
  to $\ltstokripke(p)\sat \somepath\ f_{1}\ \until\ f_{2}$, as
  desired.  The $\theta; \pass$ choices play the same role as in the
  $\somepath\globally$ case (namely, they account for the end of a
  path since they can only be taken when no other action is
  available).

  Thus we complete the proof and the conversion between CTL formulae
  and sequential tests.  Indeed, note that $\somepath\xnext$,
  $\somepath\eventually$, $\somepath\globally$, and $\somepath\until$
  is a complete, minimal set of temporal operators for CTL
  \cite{clarke99}, so all the remaining CTL constructs can be
  rewritten using only the constructs discussed above. 
\end{proof}

\begin{qexample}{How to test that your coffee machine is working}
  \label{ex-coffee-1}
  We turn our attention again to the coffee machines $b_{1}$ and
  $b_{2}$ below, as presented earlier in Example~\ref{ex-coffee}, and
  also graphically as LTS in Figure~\ref{fig-coffee}:
  \begin{eqnarray*}
    b_{1} &=& \coin; (\tea\ \choice\ \bang; \coffee)\quad \choice
    \quad \coin; (\coffee\ \choice\ \bang; \tea)\\
    b_{2} &=& \coin; (\tea\ \choice\ \bang; \tea)\quad \choice
    \quad \coin; (\coffee\ \choice\ \bang; \coffee)
  \end{eqnarray*}
  Also recall that the following CTL formula was found to
  differentiate between the two machines:
  \[\phi = \coin \land \somepath \xnext\
  (\coffee \lor \lnot \coffee \land \bang \land \somepath\xnext\
  \coffee)\]
  Indeed, the formula holds for both the initial states of
  $\ltstokripke(b_{1})$ (where coffee is offered from the outset or follows a
  hit on the machine) but holds in only one of the initial states of
  $\ltstokripke(b_{2})$ (the one that dispenses coffee).

  Let $C = \{\coin, \tea, \bang, \coffee\}$ be the set of all actions.
  We will henceforth convert silently all the tests $\theta; \pass$
  into $\pass$ as long as these tests do not participate in a choice.
  We will also perform simplifications of the intermediate tests as we
  go along in order to simplify (and so clarify) the presentation.

  The process of converting the formula $\phi$ to a test suite
  $\ctltoft_{\ltstokripke}(\phi)$ then goes as follows:
  \begin{enumerate}
  \item We convert first $\bang \land \somepath\xnext\ \coffee = \neg
    (\neg \bang \vee \neg \somepath\xnext \ \coffee)$.

    We have $\ctltoft_{\ltstokripke}(\neg\bang) = \bang; \pstop\ \choice\ \theta;
    \pass$.  On the other hand $\ctltoft_{\ltstokripke}(\somepath\xnext\ \coffee) =
    \mchoice\{a; \coffee; \pass: a\in C\}$ and so $\ctltoft_{\ltstokripke}(\neg
    \somepath\xnext\ \coffee) = \mchoice\{a; \coffee; \pstop\ \choice\
    \theta; \pass: a\in C\}\ \choice\ \theta; \pass = \mchoice\{a;
    \coffee; \pstop\ \choice\ \theta; \pass: a\in C\}$ (we ignore the
    topmost $\theta$ branch since the rest of the choice covers all
    the possible actions).

    Now we compute the disjunction $\ctltoft_{\ltstokripke}(\neg\bang)\lor
    \ctltoft_{\ltstokripke}(\neg \somepath\xnext\ \coffee)$.  With the notations used
    in the proof of Lemma~\ref{th-ft-or} we have $B_{1} = \{\bang\}$,
    $B_{2} = C$, $t_{1}(\bang) = \pstop$, $t_{N1} = \pass$, $t_{2}(b)
    = \coffee; \pstop\ \choice\ \theta; \pass$ for all $b\in B_{2}$,
    and $t_{N2} = \pass$.  We therefore have
    $\ctltoft_{\ltstokripke}(\neg\coffee)\lor \ctltoft_{\ltstokripke}(\neg \somepath\xnext\ \coffee)
    = \bang; (\coffee; \pstop\ \choice\ \theta; \pass)\ \choice\
    \mchoice\{a; \pass : a\in C\setminus\{\bang\}\}\ \choice\ \theta;
    \pass$.  Therefore:
    \begin{eqnarray}
      \ctltoft_{\ltstokripke}(\neg(\bang \land \somepath\xnext\ \coffee)) &=& \bang; (\coffee;
      \pstop\ \choice\ \theta; \pass) \nonumber \\
      && \choice\ \theta; \pass \label{eq-ex3-1}
    \end{eqnarray}
    For brevity we integrated the $C\setminus\{\bang\}$ into the
    $\theta$ branch.  Negating this test yields:
    \begin{eqnarray*}
      \ctltoft_{\ltstokripke}(\bang \land \somepath\xnext\ \coffee) &=& \bang; \coffee; \pass
    \end{eqnarray*}
    Note in passing that that the negated version as shown in
    Equation~(\ref{eq-ex3-1}) will suffice.  The last formula is only
    provided for completeness and also as a checkpoint in the
    conversion process.  Indeed, the equivalence between $\bang \land
    \somepath\xnext\ \coffee$ and $\bang; \coffee; \pass$ can be
    readily ascertained intuitively.

  \item We move to $\neg \coffee \land \bang \land \somepath\xnext\
    \coffee = \neg(\coffee \lor \neg (\bang \land \somepath\xnext\
    \coffee))$.  By Equation~(\ref{eq-ex3-1}) we have $\ctltoft_{\ltstokripke}(\coffee
    \lor \neg (\bang \land \somepath\xnext\ \coffee)) = \coffee; \pass
    \lor (\bang; (\coffee; \pstop\ \choice\ \theta; \pass)\ \choice\
    \theta; \pass)$.  This time $B_{1}=\{\coffee\}$,
    $B_{2}=\{\bang\}$, $t_{1}(b) = \pass$, $t_{2}(b) = \coffee;
    \pstop\ \choice\ \theta; \pass$, $t_{N1} = \pstop$, and $t_{N2} =
    \pass$.  Therefore $\ctltoft_{\ltstokripke}(\coffee \lor \neg (\bang \land
    \somepath\xnext\ \coffee)) = \coffee; \pass\ \choice\ \bang;
    (\coffee; \pstop\ \choice\ \theta; \pass)\ \choice\ \theta; \pass$
    (note that $B_{1}\cap B_{2}=\emptyset$).  Negating this test
    yields:
    \begin{eqnarray}
      \ctltoft_{\ltstokripke}(\neg \coffee \land \bang \land \somepath\xnext\ \coffee)
      &=&
      \coffee; \pstop\ \choice \nonumber\\
      && \bang; \coffee; \pass \label{eq-ex3-2}
    \end{eqnarray}
    This is yet another checkpoint in the conversion, as the
    equivalence above can be once more easily ascertained.

  \item The conversion of $\coffee \lor \neg \coffee \land \bang \land
    \somepath\xnext\ \coffee$ combines in a disjunction the test
    $\coffee; \pass$ and the test from Equation~(\ref{eq-ex3-2}).
    None of these tests feature a $\theta$ branch in their top choice
    and so the combination is a simple choice between the two:
    \begin{eqnarray}
      & \ctltoft_{\ltstokripke}(\coffee \lor \neg \coffee \land \bang \land \somepath\xnext\ \coffee)
      \nonumber \\
      & \qquad = \quad \coffee; \pass\ \choice\ \bang; \coffee; \pass  \label{eq-ex3-3}
    \end{eqnarray}
    In what follows we use for brevity $\phi' = \coffee \lor \neg
    \coffee \land \bang \land \somepath\xnext\ \coffee$.

  \item We have $\ctltoft_{\ltstokripke}(\somepath\xnext\ \phi') = \mchoice\{a;
    \ctltoft_{\ltstokripke}(\phi') : a\in C\}$ and therefore
    \begin{eqnarray}
      \label{eq-ex3-4a}
      \ctltoft_{\ltstokripke}(\somepath\xnext\ \phi')
      &=&
      \mchoice\{a; (\coffee; \pass \nonumber \\
      && \qquad \quad \choice\ \bang; \coffee; \pass): a\in C\}
    \end{eqnarray}
    We will actually need in what follows the negation of this
    formula, which is the following:
    \begin{eqnarray}
    \ctltoft_{\ltstokripke}(\neg \somepath\xnext\ \phi')
    &=&
    \mchoice\{a; (\coffee; \pstop \nonumber\\
    && \qquad \quad \choice\ \bang; (\coffee; \pstop\
    \choice\ \theta; \pass) \nonumber\\
    && \qquad \quad \choice\ \theta; \pass): a\in C\} \nonumber\\
    && \choice\ \theta; \pass
    \label{eq-ex3-4}
    \end{eqnarray}
    At this point our test becomes complex enough so that we can use
    it to illustrate in more detail the negation algorithm
    (Lemma~\ref{th-ft-neg}).  We thus take this opportunity to explain
    in detail the conversion between $\ctltoft_{\ltstokripke}(\somepath\xnext\
    \phi')$ from Equation~(\ref{eq-ex3-4a}) and
    $\nottest{\ctltoft_{\ltstokripke}(\somepath\xnext\ \phi')} = \ctltoft_{\ltstokripke}(\neg
    \somepath\xnext\ \phi')$ shown in Equation~(\ref{eq-ex3-4}).

    \begin{figure}
      \centering
      \begin{tabular}{ccc}
        \mbox{\begin{picture}(28,44)(-5,-2)
            \gasset{Nframe=n,ExtNL=y,NLdist=1,NLangle=180,Nw=2,Nh=2,Nframe=n,Nfill=y}

            \node[NLangle=0](t1)(5,40){$t_{1} = \ctltoft_{\ltstokripke}(\somepath\xnext\ \phi')$}
            \node(t2)(5,30){$t_{2}$}
            \node(t3)(0,20){$t_{3}$}
            \node(t4)(10,20){$t_{4}$}
            \node(t5)(0,10){$t_{5}$}
            \node(t6)(10,10){$t_{6}$}
            \node(t7)(10,0){$t_{7}$}

            \drawtrans[l](t1,t2){\small $a\in C$}
            \drawtrans[r](t2,t3){\small coffee}
            \drawtrans[l](t2,t4){\small bang}
            \drawtrans[r](t3,t5){\small $\gamma$}
            \drawtrans[l](t4,t6){\small coffee}
            \drawtrans[l](t6,t7){\small $\gamma$}
          \end{picture}}
        & \qquad &
        \mbox{\begin{picture}(36,44)(-10,-2)
            \gasset{Nframe=n,ExtNL=y,NLdist=1,NLangle=180,Nw=2,Nh=2,Nframe=n,Nfill=y}

            \node[NLangle=0](t1)(5,40){$t_{1} = \nottest{\ctltoft_{\ltstokripke}(\somepath\xnext\ \phi')}$}
            \node(t2)(5,30){$t_{2}$}
            \node(t3)(0,20){$t_{3}$}
            \node(t4)(10,20){$t_{4}$}
            \node(t6)(10,10){$t_{6}$}
            \node(t8)(-5,40){$t_{8}$}
            \node(tt1)(-5,30){}
            \node[NLangle=0](t9)(20,30){$t_{9}$}
            \node(tt2)(20,20){}
            \node[NLangle=0](t10)(20,15){$t_{10}$}
            \node(tt3)(20,5){}

            \drawtrans[l](t1,t2){\small $a\in C$}
            \drawtrans[r](t2,t3){\small coffee}
            \drawtrans[l](t2,t4){\small bang}
            \drawtrans[r](t4,t6){\small coffee}
            \drawtrans[l](t1,t8){\small $\theta$}
            \drawtrans[l](t8,tt1){\small $\gamma$}
            \drawtrans[l](t2,t9){\small $\theta$}
            \drawtrans[l](t9,tt2){\small $\gamma$}
            \drawtrans[l](t4,t10){\small $\theta$}
            \drawtrans[l](t10,tt3){\small $\gamma$}
          \end{picture}}
        \\
        $(a)$  & \qquad & $(b)$
      \end{tabular}
      \caption{Conversion of the test $\ctltoft_{\ltstokripke}(\somepath\xnext\
        \phi')$ $(a)$ into its negation
        $\nottest{\ctltoft_{\ltstokripke}(\somepath\xnext\ \phi')}$ $(b)$.}
      \label{fig-ex3-neg1}
    \end{figure}

    \begin{enumerate}
    \item The test $\ctltoft_{\ltstokripke}(\somepath\xnext\ \phi')$ is shown as an
      LTS in Figure~\ref{fig-ex3-neg1}$(a)$.  For convenience all the
      states are labeled $t_{i}$, $1\leq i\leq 7$ so that we can
      easily refer to them.

    \item We then eliminate all the success ($\gamma$) transitions.
      The states $t_{3}$ and $t_{6}$ are thus converted from \pass\ to
      \pstop.

    \item All the states that were not converted in the previous step
      (that is, states $t_{1}$, $t_{2}$, and $t_{4}$) gain a $\theta;
      \pass$ branch.  The result is the negation of the original test
      and is shown in Figure~\ref{fig-ex3-neg1}$(b)$.

    \item If needed, the conversion the other way around would proceed
      as follows: The success transitions are eliminated (this affects
      $t_{8}$, $t_{9}$, and $t_{10}$).  The preceding $\theta$
      transitions are also eliminated (which eliminates $t_{8}$,
      $t_{9}$, $t_{10}$ and affects the states $t_{1}$, $t_{2}$, and
      $t_{4}$).  The states unaffected by this process are $t_{3}$ and
      $t_{6}$ so they both gain a $\theta; \pass$ branch; however,
      such a branch is the only outgoing one for both $t_{3}$ and
      $t_{6}$ so it is equivalent to a simple $\pass$ in both cases.
      The result is precisely the test $\ctltoft_{\ltstokripke}(\somepath\xnext\
      \phi')$ as shown in Figure~\ref{fig-ex3-neg1}$(a)$.
    \end{enumerate}

  \item We finally reach the top formula $\phi$.  Indeed, $\phi =
    \coin \land \somepath\xnext\ \phi' = \neg (\neg \coin\lor \neg
    \somepath\xnext\ \phi')$.  We thus need to combine in a
    disjunction the test $\coin; \pstop\ \choice\ \theta; \pass$ with
    the test shown in Equation~(\ref{eq-ex3-4}).  We have:
    \begin{eqnarray*}
      \ctltoft_{\ltstokripke}(\neg \coffee\lor \neg \somepath\xnext\
      \phi')
      &=& \coin; (\coffee; \pstop\ \\
      && \qquad \quad \choice\ \bang; (\coffee;
      \pstop\ \choice\ \theta; \pass) \\
      && \qquad \quad \choice\ \theta; \pass)\\
      && \choice\ \mchoice\{b; \pass : b\in C\setminus\{\coin\}\} \\
      && \choice\ \theta; \pass
    \end{eqnarray*}
    Indeed, $B_{1}=\{\coin\}$, $B_{2}= C$, $t_{1}(b) = \pstop$,
    $t_{2}(b) = \coffee; \pstop\ \choice\ \bang; (\coffee; \pstop\
    \choice\ \theta; \pass)\ \choice\ \theta; \pass$, and $t_{N1} =
    t_{N2} = \pass$.

    To reduce the size of the expression we combine the
    $C\setminus\{\coin\}$ and $\theta$ choices and so we obtain:
    \begin{eqnarray*}
    \ctltoft_{\ltstokripke}(\neg \coffee\lor \neg \somepath\xnext\ \phi')
    &=&
    \coin; (\coffee; \pstop\\
    && \qquad \quad \choice\ \bang; (\coffee; \pstop\ \choice\
    \theta; \pass)\\
    && \qquad \quad \choice\ \theta; \pass)\\
    && \choice\ \theta; \pass
    \end{eqnarray*}
    Negating the above expression results in the test equivalent to
    the original formula:
    \begin{eqnarray*}
      \ctltoft_{\ltstokripke}(\phi) &=&
      \coin; (\coffee; \pass\ \choice\ \bang; \coffee; \pass)
    \end{eqnarray*}
  \end{enumerate}
  Recall now that the test we started from in Example~\ref{ex-coffee}
  was slightly different, namely:
  \[
  t= \coin; (\coffee; \pass\ \choice\ \theta; \bang; \coffee; \pass)
  \]
  We argue however that these two tests are in this case equivalent.
  Indeed, both tests succeed whenever \coin\ is followed by \coffee.
  Suppose now that \coin\ does happen but the next action is not
  \coffee.  Then $t$ will follow on the deadlock detection branch,
  which will only succeed if the next action is \bang.  On the other
  hand $\ctltoft_{\ltstokripke}(\phi)$ does not have a deadlock detection branch in the
  choice following \coin; however, the only alternative to \coffee\ in
  $\ctltoft_{\ltstokripke}(\phi)$ is \bang, which is precisely the same alternative as
  for $t$ (as shown above).  We thus conclude that $t$ and
  $\ctltoft_{\ltstokripke}(\phi)$ are indeed equivalent.
\end{qexample}

\begin{lemma}
  \label{th-ctl-ft-x}
  There exist a function
  $\ctltoft_{\ltstokripkex}:\ctlset\rightarrow\ftset$ such that
  $\ltstokripkex(p)\models f$ iff
  $p\ \may\ \ctltoft_{\ltstokripkex}(f)$ for any $p\in\procset$.
\end{lemma}

\begin{proof}
  The proof established earlier for $\ctltoft_{\ltstokripke}$
  (Lemma~\ref{th-ctl-ft}) will also work for
  $\ctltoft_{\ltstokripkex}$.  Indeed, the way the operator $\models$
  is defined (Definition~\ref{def-Kripke-Kripke1-sat}) ensures that
  all occurrences of $\Delta$ are ``skipped over'' as if they were not
  there in the first place.  However, the paths without the $\Delta$
  labels are identical to the paths examined in Lemma~\ref{th-ctl-ft}. 
\end{proof}

\section{Conclusions}
\label{sect-conclusions}

Our work creates a constructive equivalence between CTL and failure
trace testing.  For this purpose we start by offering an equivalence
relation between a process or a labeled transition system and a Kripke
structure (Section~\ref{sect-lts-kripke-a}), then the development of
an algorithmic function $\ltstokripke$ that converts any labeled
transition system into an equivalent Kripke structure
(Theorem~\ref{th-kripke-lts}).

As already mentioned, the function $\ltstokripke$ creates a Kripke
structure that may have multiple initial states, and so a weaker
satisfaction operator (over sets of states rather than states) was
needed.  It was further noted that this issue only happens for those
LTS that start with a choice of multiple visible actions.  It follows
that in order to eliminate the need for a new satisfaction operator
(over sets of states) one can in principle simply create an extra LTS
state which becomes the initial state and performs an artificial
``start'' action to give control to the original initial state.  This
claim needs however to be verified.  Alternatively, one can
investigate another equivalence relation (and subsequent conversion
function) without this disadvantage. We offer precisely such a
relation (Section~\ref{sect-lts-kripke-a}) and the corresponding
algorithmic conversion function $\ltstokripkex$
(Theorem~\ref{new-th-kripke-lts}).

There are advantages and disadvantages to both these approaches.  The
conversion function $\ltstokripke$ uses an algorithmic conversion
between LTS and Kripke structures that results in very compact Kripke
structures but introduces the need to modify the model checking
algorithm (by requiring a modified notion of satisfaction for CTL
formulae).  The function $\ltstokripke$ on the other hand results in
considerably larger Kripke structures but does not require any
modification of the model checking algorithm.

Once the conversion functions are in place we develop algorithmic
functions for the conversion of failure trace tests to equivalent CTL
formulae and the other way around, thus showing that CTL and failure
trace testing are equivalent (Theorem~\ref{th-ctl-ft-main}).
Furthermore this equivalence holds for both notions of equivalence
between LTS and Kripke structure and so our thesis is that failure
trace testing and CTL are equivalent under any reasonable equivalence
relation between LTS and Kripke structures.

Finally, we note that the straightforward, inductive conversion of
failure trace tests into CTL formulae produces in certain cases
infinite formulae.  We address this issue by showing how failure trace
tests containing cycles (which produce the unacceptable infinite
formulae) can be converted into compact (and certainly finite) CTL
formulae (Theorem~\ref{th-ctl-ftr-compt}).

We believe that our results (providing a combined, logical and
algebraic method of system verification) have unquestionable
advantages.  To emphasize this point consider the scenario of a
network communication protocol between two end points and through some
communication medium being formally specified.  The two end points are
likely to be algorithmic (or even finite state machines) and so the
natural way of specifying them is algebraic.  The communication medium
on the other hand has a far more loose specification.  Indeed, it is
likely that not even the actual properties are fully known at
specification time, since they can vary widely when the protocol is
actually deployed (between say, the properties of a 6-foot direct
Ethernet link and the properties of a nondeterministic and congested
Internet route between Afghanistan and Zimbabwe).  The properties of
the communication medium are therefore more suitable to logic
specification.  Such a scenario is also applicable to systems with
components at different levels of maturity (some being fully
implemented already while others being at the prototype stage of even
not being implemented at all and so less suitable for being specified
algebraically).  Our work enables precisely this kind of mixed
specification.  In fact no matter how the system is specified we
enable the application of either model checking of model-based testing
(or even both) on it, depending on suitability or even personal taste.

The results of this paper are important first steps towards the
ambitious goal of a unified (logic and algebraic) approach to
conformance testing. We believe in particular that this paper opens
several direction of future research.

The issue of tests taking an infinite time to complete is an
ever-present issue in model-based testing.  Our conversion of CTL
formulae is no exception, as the tests resulting from the conversion
of expressions that use $\somepath\eventually$, $\somepath\globally$,
and $\somepath\until$ fall all into this category.  Furthermore Rice's
theorem \cite{papa} (which states that any non-trivial and extensional
property of programs is undecidable) guarantees that tests that take
an infinite time to complete will continue to exist no matter how much
we refine our conversion algorithms.  We therefore believe that it is
very useful to investigate methods and algorithms for partial (or
incremental) application of tests.  Such methods will offer
increasingly stronger guarantees of correctness as the test
progresses, and total correctness at the limit (when the test
completes).

It would also be interesting to extend this work to other temporal
logics (such as CTL*) and whatever testing framework turns out to be
equivalent to it (in the same sense as used in our work).

\bibliographystyle{siam}
\bibliography{formal,algos}

\begin{thebibliography}{10}

\bibitem{alur94}
{\sc R.~Alur and D.~L. Dill}, {\em A theory of timed automata}, Theoretical
  Computer Science, 126 (1994), pp.~183--235.

\bibitem{bellini00}
{\sc P.~Bellini, R.~Mattolini, and P.~Nesi}, {\em Temporal logics for real-time
  system specification}, {ACM} Computing Surveys, 32 (2000), pp.~12--42.

\bibitem{brinksma87}
{\sc E.~Brinksma, G.~Scollo, and C.~Steenbergen}, {\em {LOTOS} specifications,
  their implementations and their tests}, in {IFIP} 6.1 Proceedings, 1987,
  pp.~349--360.

\bibitem{motres05}
{\sc M.~Broy, B.~Jonsson, J.-P. Katoen, M.~Leucker, and A.~Pretschner}, eds.,
  {\em Model-Based Testing of Reactive Systems: Advanced Lectures}, vol.~3472
  of Lecture Notes in Computer Science, Springer, 2005.

\bibitem{bruda05}
{\sc S.~D. Bruda}, {\em Preorder relations}, in Broy et~al. \cite{motres05},
  pp.~117--149.

\bibitem{bruda10j}
{\sc S.~D. Bruda and C.~Dai}, {\em A testing theory for real-time systems},
  International Journal of Computers, 4 (2010), pp.~97--106.

\bibitem{bruda09a}
{\sc S.~D. Bruda and Z.~Zhang}, {\em Refinement is model checking: From failure
  trace tests to computation tree logic}, in Proceedings of the 13th {IASTED}
  International Conference on Software Engineering and Applications (SEA 09),
  Cambridge, MA, Nov. 2009.

\bibitem{clarke82}
{\sc E.~M. Clarke and E.~A. Emerson}, {\em Design and synthesis of
  synchronization skeletons using branching-time temporal logic}, in Works in
  Logic of Programs, 1982, pp.~52--71.

\bibitem{clarke86}
{\sc E.~M. Clarke, E.~A. Emerson, and A.~P. Sistla}, {\em Automatic
  verification of finite state concurrent systems using temporal logic
  specification}, ACM Transactions on Programming Languages and Systems, 8
  (1986), pp.~244--263.

\bibitem{clarke99}
{\sc E.~M. Clarke, O.~Grumberg, and D.~A. Peled}, {\em Model Checking}, MIT
  Press, 1999.

\bibitem{cleaveland00}
{\sc R.~Cleaveland and G.~L\"{u}ttgen}, {\em Model checking is
  refinement---{R}elating {B}\"{u}chi testing and linear-time temporal logic},
  Tech. Rep. 2000-14, ICASE, Langley Research Center, Hampton, VA, Mar. 2000.

\bibitem{dai08}
{\sc C.~Dai and S.~D. Bruda}, {\em A testing framework for real-time
  specifications}, in Proceedings of the 9th IASTED International Conference on
  Software Engineering and Applications (SEA 08), Orlando, Florida, Nov. 2008.

\bibitem{denicola84}
{\sc R.~{De Nicola} and M.~C.~B. Hennessy}, {\em Testing equivalences for
  processes}, Theoretical Computer Science, 34 (1984), pp.~83--133.

\bibitem{nicola90}
{\sc R.~{De Nicola} and F.~Vaandrager}, {\em Action versus state based logics
  for transition systems}, Lecture Notes in Computer Science, 469 (1990),
  pp.~407--419.

\bibitem{nicola95}
\leavevmode\vrule height 2pt depth -1.6pt width 23pt, {\em Three logics for
  branching bisimulation}, Journal of the ACM, 42 (1995), pp.~438--487.

\bibitem{floyd67}
{\sc R.~W. Floyd}, {\em Assigning meanings to programs}, in Mathematical
  Aspects of Computer Science, J.~T. Schwartz, ed., vol.~19 of Proceedings of
  Symposia in Applied Mathematics, American Mathematical Society, 1967,
  pp.~19--32.

\bibitem{gerth95}
{\sc R.~Gerth, D.~Peled, M.~Y. Vardi, and P.~Wolper}, {\em Simple on-the-fly
  automatic verification of linear temporal logic}, in Proceedings of the
  {IFIP} symposium on Protocol Specification, Testing and Verification ({PSTV}
  95), Warsaw, Poland, 1995, pp.~3--18.

\bibitem{hoare69}
{\sc C.~A.~R. Hoare}, {\em An axiomatic basis for computer programming},
  Communications of the ACM, 12 (1969), pp.~576--580 and 583.

\bibitem{katoen05}
{\sc J.-P. Katoen}, {\em Labelled transition systems}, in Broy et~al.
  \cite{motres05}, pp.~615--616.

\bibitem{langerak89}
{\sc R.~Langerak}, {\em A testing theory for {LOTOS} using deadlock detection},
  in Proceedings of the {IFIP} {WG6.1} Ninth International Symposium on
  Protocol Specification, Testing and Verification {IX}, 1989, pp.~87--98.

\bibitem{papa}
{\sc H.~R. Lewis and C.~H. Papadimitriou}, {\em Elements of the Theory of
  Computation}, Prentice-Hall, 2nd~ed., 1998.

\bibitem{pawlikowski90}
{\sc K.~Pawlikowski}, {\em Steady-state simulation of queueing processes:
  survey of problems and solutions}, {ACM} Computing Surveys, 22 (1990),
  pp.~123--170.

\bibitem{pnueli81}
{\sc A.~Pnueli}, {\em A temporal logic of concurrent programs}, Theoretical
  Computer Science, 13 (1981), pp.~45--60.

\bibitem{queille83}
{\sc J.~P. Queille and J.~Sifakis}, {\em Fairness and related properties in
  transition systems --- a temporal logic to deal with fairness}, Acta
  Informatica, 19 (1983), pp.~195--220.

\bibitem{saidi97}
{\sc H.~Sa\"{\i}di}, {\em The invariant checker: Automated deductive
  verification of reactive systems}, in Proceedings of Computer Aided
  Verification (CAV 97), vol.~1254 of Lecture Notes In Computer Science,
  Springer, 1997, pp.~436--439.

\bibitem{schneider00}
{\sc S.~Schneider}, {\em Concurrent and Real-time Systems: The CSP Approach},
  John Wiley \& Sons, 2000.

\bibitem{schriber03}
{\sc T.~J. Schriber, J.~Banks, A.~F. Seila, I.~St{\aa}hl, A.~M. Law, and R.~G.
  Born}, {\em Simulation textbooks - old and new, panel}, in Winter Simulation
  Conference, 2003, pp.~1952--1963.

\bibitem{thomas90}
{\sc W.~Thomas}, {\em Automata on infinite objects}, in Handbook of Theoretical
  Computer Science, J.~van Leeuwen, ed., vol.~B, North Holland, 1990,
  pp.~133--191.

\bibitem{tretmans96}
{\sc J.~Tretmans}, {\em Conformance testing with labelled transition systems:
  Implementation relations and test generation}, Computer Networks and {ISDN}
  Systems, 29 (1996), pp.~49--79.

\bibitem{vardi86}
{\sc M.~Vardi and P.~Wolper}, {\em An automata-theoretic approach to automatic
  program verification}, in Proceedings of the First Annual Symposium on Logic
  in Computer Science ({LICS} 86), 1986, pp.~332--344.

\bibitem{vardi94}
{\sc M.~Y. Vardi and P.~Wolper}, {\em Reasoning about Infinite Computations},
  Academic Press, 1994.

\end{thebibliography}

\end{document}